\PassOptionsToPackage{unicode}{hyperref}
\PassOptionsToPackage{hyphens}{url}
\PassOptionsToPackage{dvipsnames,svgnames,x11names}{xcolor}
\documentclass[
  12pt]{article}

\usepackage{graphicx}
\usepackage{float}
\restylefloat{table}
\usepackage{subcaption}
\usepackage{array, booktabs, makecell, multirow, longtable}
\usepackage{amsmath, amssymb, amsthm}   
\usepackage{mathtools}                  
\usepackage{mathrsfs}                  
\usepackage{bm}              
\usepackage{makecell}
\usepackage{esvect}                     
\usepackage{cancel}                     
\usepackage{latexsym, amsfonts, amstext}
\usepackage{enumitem}
\newtheorem{theorem}{Theorem}
\newtheorem{lemma}{Lemma}
\newtheorem{procedure}{Procedure}

\newtheorem{proposition}{Proposition}
\newtheorem{remark}{Remark}
\newtheorem{corollary}{Corollary}
\newtheorem{algorithm}{Algorithm}
\usepackage{iftex}
\ifPDFTeX
  \usepackage[T1]{fontenc}
  \usepackage[utf8]{inputenc}
  \usepackage{textcomp} % provide euro and other symbols
\else % if luatex or xetex
  \usepackage{unicode-math}
  \defaultfontfeatures{Scale=MatchLowercase}
  \defaultfontfeatures[\rmfamily]{Ligatures=TeX,Scale=1}
\fi
\usepackage{lmodern}
\ifPDFTeX\else  
\fi
\IfFileExists{upquote.sty}{\usepackage{upquote}}{}
\IfFileExists{microtype.sty}{% use microtype if available
  \usepackage[]{microtype}
  \UseMicrotypeSet[protrusion]{basicmath} % disable protrusion for tt fonts
}{}
\makeatletter
\@ifundefined{KOMAClassName}{% if non-KOMA class
  \IfFileExists{parskip.sty}{%
    \usepackage{parskip}
  }{% else
    \setlength{\parindent}{0pt}
    \setlength{\parskip}{6pt plus 2pt minus 1pt}}
}{% if KOMA class
  \KOMAoptions{parskip=half}}
\makeatother
\usepackage{xcolor}
\makeatletter
\ifx\paragraph\undefined\else
  \let\oldparagraph\paragraph
  \renewcommand{\paragraph}{
    \@ifstar
      \xxxParagraphStar
      \xxxParagraphNoStar
  }
  \newcommand{\xxxParagraphStar}[1]{\oldparagraph*{#1}\mbox{}}
  \newcommand{\xxxParagraphNoStar}[1]{\oldparagraph{#1}\mbox{}}
\fi
\ifx\subparagraph\undefined\else
  \let\oldsubparagraph\subparagraph
  \renewcommand{\subparagraph}{
    \@ifstar
      \xxxSubParagraphStar
      \xxxSubParagraphNoStar
  }
  \newcommand{\xxxSubParagraphStar}[1]{\oldsubparagraph*{#1}\mbox{}}
  \newcommand{\xxxSubParagraphNoStar}[1]{\oldsubparagraph{#1}\mbox{}}
\fi
\makeatother

\usepackage{longtable,booktabs,array}
\usepackage{calc} % for calculating minipage widths
\usepackage{etoolbox}
\makeatletter
\patchcmd\longtable{\par}{\if@noskipsec\mbox{}\fi\par}{}{}
\makeatother
\IfFileExists{footnotehyper.sty}{\usepackage{footnotehyper}}{\usepackage{footnote}}
\makesavenoteenv{longtable}
\usepackage{graphicx}
\makeatletter
\def\maxwidth{\ifdim\Gin@nat@width>\linewidth\linewidth\else\Gin@nat@width\fi}
\def\maxheight{\ifdim\Gin@nat@height>\textheight\textheight\else\Gin@nat@height\fi}
\makeatother
\setkeys{Gin}{width=\maxwidth,height=\maxheight,keepaspectratio}
\makeatletter
\def\fps@figure{htbp}
\makeatother

\makeatletter
\@ifpackageloaded{caption}{}{\usepackage{caption}}
\AtBeginDocument{%
\ifdefined\contentsname
  \renewcommand*\contentsname{Table of contents}
\else
  \newcommand\contentsname{Table of contents}
\fi
\ifdefined\listfigurename
  \renewcommand*\listfigurename{List of Figures}
\else
  \newcommand\listfigurename{List of Figures}
\fi
\ifdefined\listtablename
  \renewcommand*\listtablename{List of Tables}
\else
  \newcommand\listtablename{List of Tables}
\fi
\ifdefined\figurename
  \renewcommand*\figurename{Figure}
\else
  \newcommand\figurename{Figure}
\fi
\ifdefined\tablename
  \renewcommand*\tablename{Table}
\else
  \newcommand\tablename{Table}
\fi
}
\@ifpackageloaded{float}{}{\usepackage{float}}
\floatstyle{ruled}
\@ifundefined{c@chapter}{\newfloat{codelisting}{h}{lop}}{\newfloat{codelisting}{h}{lop}[chapter]}
\floatname{codelisting}{Listing}

\makeatother
\makeatletter
\@ifpackageloaded{caption}{}{\usepackage{caption}}
\@ifpackageloaded{subcaption}{}{\usepackage{subcaption}}
\makeatother

\ifLuaTeX
  \usepackage{selnolig}  % disable illegal ligatures
\fi
\usepackage[]{natbib}
\usepackage{bookmark}

\IfFileExists{xurl.sty}{\usepackage{xurl}}{} % add URL line breaks if available
\hypersetup{
  pdftitle={Title},
  pdfauthor={Author 1; Author 2},
  pdfkeywords={3 to 6 keywords, that do not appear in the title},
  colorlinks=true,
  linkcolor={blue},
  filecolor={Maroon},
  citecolor={Blue},
  urlcolor={Blue},
  pdfcreator={LaTeX via pandoc}}

\newcommand{\anon}{1}

\begin{document}

\def\spacingset#1{\renewcommand{\baselinestretch}%
{#1}\small\normalsize} \spacingset{1}

%%%%%%%%%%%%%%%%%%%%%%%%%%%%%%%%%%%%%%%%%%%%%%%%%%%%%%%%%%%%%%%%%%%%%%%%%%%%%%

\if1\anon
{
  \title{\bf Equivalence Test for Mean Functions from Multi-population Functional Data}
  \author{Chuang Xu\thanks{The authors are grateful to the Australian Research Council for supporting this research through grant DP220102232\hspace{.2cm}},\, \, Andrew T. A. Wood and  Yanrong Yang\\
    Research School of Finance, Actuarial Studies and Statistics, \\
    Australian National University }
   % and \\
   % Author 2 \\
   % Department of ZZZ, University of WWW}
  \maketitle
} \fi

\if0\anon
{
  \bigskip
  \bigskip
  \bigskip
  \begin{center}
    {\LARGE\bf Title}
\end{center}
  \medskip
} \fi

\bigskip
\begin{abstract}
Most existing methods for testing equality of means of functional data from multiple populations rely on assumptions of equal covariance and/or Gaussianity. In this work we provide a new testing method based on a statistic that is distribution-free under the null hypothesis (i.e. the statistic is pivotal), and allows different covariance structures across populations, while Gaussianity is not required. In contrast to classical methods of functional mean testing, where either observations of the full curves or projections are applied, our method allows the projection dimension to increase with the sample size to allow asymptotic recovery of full information as the sample size increases. We obtain a unified theory for the asymptotic distribution of the test statistic under local alternatives, in both the sample and bootstrap cases. The finite sample performance for both size and power have been studied via simulations and the approach has also been applied to two real datasets.  
\end{abstract}

\noindent%
{\it Keywords:} Asymptotic normality; Consistency; Functional data; $k$-sample problems.
\vfill

\newpage
\spacingset{1.8} % DON'T change the spacing!

\section{Introduction}\label{sec-intro}

\subsection{Statement of the problem}

Functional data arise from observations of phenomena that are continuous over time or space, and their analysis is crucial in fields such as biology, chemistry, economics, engineering, environmental science, and medicine.  For further details and applications of functional data analysis, we refer to the well-known monographs \cite{Ramsay2005,Ferraty2006} and, for theoretical foundations, see \cite{Hsing2015}.

With complex data structures such as functional data, modern statistics aims to increase the data size by accumulating data from multiple populations and to achieve more accurate statistical inference by exploiting homogeneity across multi-populations when homogeneity is present, while also accounting for any heterogeneity in the combined data. For instance, multi-task learning (\cite{Duan2023}) and transfer learning (\cite{Cai2024}) are machine learning techniques to deal with multiple datasets collected from different sources. 

Prior information on homogeneity and heterogeneity is helpful to attain more efficient statistical inference on integrated data from multiple populations. In this article, we contribute to multi-population functional data in which we wish to test for equality of population mean functions across $k \geq 2$ populations.  In Section 5, we consider real-data applications,  one of which is concerned with human growth data.  In this setting, it is of interest to compare the mean growth curves for males and females without assuming that the covariance functions for these two populations are equal.  

The type of test required is known as a $k$-sample test, where $k \geq 2$ is the number of groups to be compared.
In this article, we propose such a test which incorporates the following desirable properties.
\begin{description}
\item[(i)] The covariance functions of the different populations are allowed to differ.
\item[(ii)] The number of populations can be any fixed $k \geq 2$.
\item[(iii)] It is not assumed that the random functions are Gaussian.  
\item[(iv)] The test statistic is (asymptotically) pivotal, i.e. the limiting distribution under the null hypothesis does not depend on unknown parameters.
\item[(v) ] The test statistic is easy to compute so that bootstrapping is feasible.
\item[(vi)] Within the class of functions considered, the test is consistent against all alternatives, in the sense that, if the null hypothesis of equal mean functions is false, then the null hypothesis is rejected with probability approaching 1 as sample sizes go to infinity.
\end{description}
To the best of our knowledge, no approaches to $k$-sample testing currently in the literature possess all of the above properties.  In the remainder of this section, we will outline the main ingredients of our approach and explain how they relate to procedures in the literature. 

%When data are collected from several populations, a primary
%question is whether the underlying mean curves are identical.
\subsection{Approaches focused  on basis function representations}
We first introduce some notation.  Let $(\Omega, \mathcal{F}, \mathbb{P})$ denote the underlying probability space, where $\Omega$ is the set of possible outcomes. For each $\omega \in \Omega$, a functional observation or process over a compact interval $\mathcal{T} \subset \mathbb{R}$ is written $X(t,\omega)$ and is a smooth function of $t \in \mathcal{T}$.
%, for each $\omega \in \Omega$. 
%where $\Omega$ is the set of possible outcomes and $(\Omega, \mathcal{F}, \mathbb{P})$ is the underlying probability space.  
It is assumed that $X(t,\omega) \in L^2(\mathcal{T})$, where $L^2(\mathcal{T})$ is endowed with usual $L^2$ inner product defined by $\langle f,g\rangle=\int_{\mathcal{T}} f(t)g(t)dt$. From now on, for simplicity we just use $X(t)$ to denote a random function when the context is clear enough, and we assume $\mathcal{T}=[0,1]$.

Let $\{b_l(t)\}_{l=1}^\infty$ denote a pre-specified basis of $L^2[0,1]$.   Broadly speaking, tests for a common mean function may be classified by the
\emph{projection dimension}~$p$—the number of basis
coefficients on which the statistic ultimately depends. First, we consider two cases.

\begin{enumerate}[label=(\alph*),leftmargin=*,itemsep=2pt]
\item \emph{Full‐curve procedures} ($p=\infty$) operate directly on
    densely sampled trajectories, integrating squared pointwise
    differences or employing global $L^{2}$ distances. These methods avoid truncation bias but typically rely on restrictive assumptions
    (e.g., common covariance) and have test statistics that yield non-pivotal limit laws.  Note that for this class of methods, we either need to calculate the functions on fine grids or to calculate the functions at arbitrary locations.

\item \emph{Fixed-dimension projections} ($p$ fixed) operates by projecting curves onto a finite collection of basis elements, after which multivariate techniques are applied.
\end{enumerate}

When a function is expanded in an orthonormal basis, the truncation error after $p$ terms typically decays at a rate that depends on the smoothness of the sample paths. 
Rough trajectories demand many basis elements for an adequate approximation, whereas very smooth trajectories can be well represented by only a few. A well-selected basis also reduces the number of basis functions needed to provide a good approximation.

The above comments suggest letting the projection dimension \emph{grow at a moderate rate with} the sample size.  From our perspective, there is an important trade-off: $p_n$, the projection dimension,  should grow at a sufficiently fast rate for the truncation error in the functional observations to be small; and $p_n$ should grow at a sufficiently slow rate to allow us to use ``fixed $p$'' methods without them breaking down.
Starting with a pre-specified orthonormal basis, for example, Fourier, B-spline, or wavelet functions, we consider the score vectors of dimension $p$ associated with them.
By letting $p_n\to\infty$ we asymptotically recover the full functional information while retaining the pivotality of the test statistics and sufficient computational simplicity to enable use of the bootstrap.
We show that in practice moderate values of $p_n$ suffice when the sample paths are reasonably smooth and the chosen basis is well aligned with the dominant features of the data.

For efficient computation and concise theory, we utilize pre-determined basis projection instead of exploring formal data-driven selection for basis functions.
Some approaches use a data-driven eigenbasis obtained from functional principal component analysis (FPCA), which is attractive for a \emph{single} population.  However, it becomes challenging when $k>1$: distinct covariance operators yield different eigenbases unless equal covariance assumption is made, the computation costs will also be quite large if $k$ is relatively large.  
Anchoring the projection on a common, pre-determined basis circumvents this difficulty and leads to tractable asymptotic results. Pre-determined bases have been used in the linear discriminant analysis literature, see \cite{Xue2024}.

%A key technical difference between inference in high-dimensional settings (e.g. \cite{Wainwright2019}) and the functional data analysis setting considered in this article is that, in the latter, the relevant covariance operators have finite trace, whereas a corresponding result does not typically hold in the limit in high-dimensional settings.  As a consequence, in our setting it is reasonable to allow $p$ to grow at a somewhat slower rate than $n$, unlike in the high-dimnesional  setting.  This allows us to employ a powerful approach to inference based which is designed for finite $p$, but still works when $p$ increases as $n$ goes to infinity in our theory. 

\subsection{Literature Review}
Testing equality of mean curves has been studied for very long time
, beginning with the two–sample, common–covariance setting of
\citet{Fan1998}, who adapted the classical $t$–statistic to functional
data by integrating pointwise tests.  Subsequent extensions to
$k>2$ groups include integrated and maximum $F$–type procedures
\citep{Cuevas2004,Zhang2014a}, permutation–based MANOVA
\citep{CuestaAlbertos2010}, and principal–component projections
\citep{Horvath2015}.  These methods rely on (at least asymptotic)
Gaussianity and homoscedastic covariances.

Heteroscedastic designs have been addressed by
\citet{JimenezGamero2021}, who proposed a self–normalising statistic
with standard–normal limits when the number of groups grows with
sample size. 

The recent work of \citet{Munko2023}  develops a stepdown procedure that controls the
family‐wise error rate across a broad class of functional hypotheses.
Although their theory permits non‐Gaussian data, it is geared toward
\emph{multiple} pairwise comparisons rather than the global equality
null considered here, moreover only diagonal information of the covariance functions were utilised in their paper.
Closely related high‐dimensional tools include
the functional discriminant analysis of \citet{Xue2024} and the
functional graphical‐model tests of \citet{Qiao2019}.

\subsection{Our contributions}
This study advances methodology for $k$–sample functional mean testing in several respects.

\begin{enumerate}[label=(\textbf{C\arabic*})]
  \item \textbf{Fully heteroscedastic, non–Gaussian framework.}  Unlike classical functional ANOVA procedures that assume identical covariance operators and (sub‐)Gaussian curves, our test accommodates \emph{group‐specific} covariance functions and only requires finite fourth moments.  This greatly broadens its applicability to biomedical, environmental, and other settings where variance heterogeneity and moderately heavy‐tailed behaviour are the rule rather than the exception.

  \item \textbf{Covariance‐adapted test statistic.}  The proposed quadratic form explicitly incorporates consistent estimators of each group’s covariance operator.  By exploiting this second‐order information, the test achieves higher local power than covariance‐agnostic competitors, as confirmed by both theory and simulations.

  \item \textbf{Unified limit theory under the null and local alternatives.}  We establish an increasing-dimension  Central Limit Theorem that covers the full sequence of local alternatives
        with $H_0$ (\(\delta_j\equiv0\)) as a special case.  The resulting Gaussian approximation provides a single theoretical umbrella for size control and local power analysis.

  \item \textbf{Bootstrap validity with negligible extra cost.}  Despite the heteroscedastic setting, a standard bootstrap remains valid under the condition that the number of scores does not grow at too fast a rate; we prove an \emph{unconditional} CLT for the bootstrapped statistic that mirrors its population counterpart where the \emph{conditional} CLT takes part in the proof.
\end{enumerate}

Collectively, these features deliver a broadly applicable testing procedure that is theoretically sharp, computationally light, and empirically powerful across a wide spectrum of functional data environments.

\subsection{Structure of the article} 
The remainder of the article is structured as follows.  
Section~2  introduces the proposed test statistic, details its computation and
presents the accompanying asymptotic theory.  In Section~3, choice of basis and choice of $p$ are discussed.   
Section~4 reports extensive Monte~Carlo studies that assess the finite‐sample behaviour of the procedure, while Section~5 illustrates its use in a real–data application.  
The statement of key technical lemmas and proofs of the main theorems  are collected in Section~6, while the proofs of the technical lemmas are given in the Supplementary Materials.

\section{Hypothesis test for k samples}
In this section, we first present the statistical model and hypothesis. Then we discuss the technical assumptions which are needed to prove the unified asymptotic theory of the test statistic and the bootstrapped version of the test statistic.  

\subsection{Preliminaries}
In this section we review the hypothesis testing problem and outline our viewpoint.\\
Let \( \{X_{j1}(t), \ldots, X_{jn_j}(t)\}_{j=1}^k \) be independent random functions, where for each group \( j = 1, \ldots, k \), the \( \{X_{ji}(t)\}_{i=1}^{n_j} \) are independent and identically distributed stochastic processes in \(  L^2(\mathcal{T}) \), where $\mathcal{T}$ is a compact interval over the real line. We denote the mean and covariance functions by
$\mu_j(t) = \mathbb{E}[X_{j1}(t)]$ and $\gamma_j(s, t) = \mathbb{E}\left[ (X_{j1}(s) - \mu_j(s))(X_{j1}(t) - \mu_j(t)) \right]$.

Our goal is to perform the following test:
\begin{equation}
    H_0: \mu_1(t)=...=\mu_k(t) \hskip 0.1truein \textrm{versus} \hskip 0.1truein
    H_1: \mu_i(t) \neq \mu_j(t), \textrm{for some $i\neq j$, $1 \leq i,j \leq k$}. \label{hypothesistest}
\end{equation}
We impose the assumption that
\begin{align}
    \mathbb{E}\left [\left (\int X_{ji}(t)^2dt \right)^4 \right] < \infty. \label{A1}
\end{align}

We project each function onto a pre-specified orthonormal basis \( \{b_l(t)\}_{l=1}^\infty \) of \( L^2[0,1] \), such as the Fourier basis or a wavelet basis. Using this basis expansion, each function admits the representation:
\begin{align*}
    X_{ji}(t) &= \sum_{l=1}^\infty y_{jil} b_l(t), \quad
    y_{jil} = \langle X_{ji}(t), b_l(t) \rangle=\int_T X_{ji}(t)b_l(t)dt.
\end{align*}

Correspondingly, the mean function can be expanded as:
$\mu_j(t) = \sum_{l=1}^\infty \mu_{jl} b_l(t)$, where  $\mu_{jl}=\mathbb{E}[y_{jil}]$.
We define the following vector notation for truncated and infinite expansions:
\begin{align}
\mu_j^{(p)} &= [\mu_{j1}, \ldots, \mu_{jp}]^\top, & \nonumber
y_{ji}^{(p)} &= [y_{ji1}, \ldots, y_{jip}]^\top, & \nonumber
\bar{y}_j^{(p)} &= \frac{1}{n_j} \sum_{i=1}^{n_j} y_{ji}^{(p)},\nonumber \\
 \mu_j^{(\infty)} &= [\mu_{j1}, \mu_{j2}, \ldots]^\top, &
y_{ji}^{(\infty)} &= [y_{ji1}, y_{ji2}, \ldots]^\top, & 
 \bar{y}_j^{(\infty)} &= \frac{1}{n_j} \sum_{i=1}^{n_j} y_{ji}^{(\infty)}. \label{ypdefinition}
\end{align}

The covariance function \( \gamma_j(s, t) \) induces a covariance operator defined by
\[
(K_j f)(t) = \int \gamma_j(s, t) f(s) \, ds, \quad \text{for } f \in L^2[0,1].
\]

We investigate the relationship between the operator \( K_j \) and the covariance matrix of the infinite-dimensional score vector \( y_j^{(\infty)} \). Since the basis is fixed, the population covariance matrix is given by 
$\Sigma_j = \operatorname{Cov}(y_{j1}^{(\infty)}) = \mathbb{E}\left[(y_{j1}^{(\infty)} - \mu_j^{(\infty)})(y_{j1}^{(\infty)} - \mu_j^{(\infty)})^\top\right]$.
Using standard calculations for $j=1, \ldots , k$ and $i=1, \ldots , n_j$,
\begin{align*}
\mathbb{E}[(y_{jil} - \mu_{jl})(y_{jir} - \mu_{jr})] 
&= \mathbb{E}\left[ \int (X_{j1}(t) - \mu_j(t)) b_l(t) \, dt \int (X_{j1}(s) - \mu_j(s)) b_r(s) \, ds \right] \\
&= \int \gamma_j(s, t) b_l(t) b_r(s) \, ds dt \\
&= \langle K_j b_l, b_r \rangle.
\end{align*}

Thus, the covariance matrix can be expressed as
$\Sigma_j = \left\{ \langle K_j b_l, b_r \rangle \right\}_{l, r = 1}^\infty$.
Because the basis is fixed, the eigenvalues and eigenvectors of the covariance operator \( K_j \) coincide with those of \( \Sigma_j \). This can be viewed as an infinite-dimensional analogue of the invariance of eigensystems under basis transformation.

The truncated \( p \times p \) covariance matrix is
$\Sigma_{j,p} = \left\{ \langle K_j b_l, b_r \rangle \right\}_{l, r = 1}^p = \operatorname{Cov}(y_{j1}^{(p)})$.
The sample covariance function is defined by:
\[
\hat{\gamma}_j(s, t) = \frac{1}{n_j} \sum_{i=1}^{n_j} \left( X_{ji}(t) - \bar{X}_j(t) \right) \left( X_{ji}(s) - \bar{X}_j(s) \right),
\]
and the associated sample covariance operator is denoted by \( \hat{K}_j \). The sample covariance matrices of the truncated scores are given by:
\begin{align}
    \hat{\Sigma}_{j,p} = \left\{ \langle \hat{K}_j b_l, b_r \rangle \right\}_{l, r = 1}^p 
= \frac{1}{n_j} \sum_{i=1}^{n_j} \left( y_{ji}^{(p)} - \bar{y}_j^{(p)} \right) \left( y_{ji}^{(p)} - \bar{y}_j^{(p)} \right)^\top, \label{sampletrunccov}
\end{align}
which is the sample covariance matrix of the truncated score vectors.
Given this formulation, we reinterpret the hypothesis test (\ref{hypothesistest}) as:
\begin{equation}
    H_0\!: \quad \mu_1^{(\infty)} = \cdots = \mu_k^{(\infty)} \hskip 0.1truein \textrm{versus} \hskip 0.1truein 
    H_1\!: \quad \mu_i^{(\infty)} \neq \mu_j^{(\infty)} \quad \text{for some } i \neq j. \label{hypothesistest2}
\end{equation}

Our goal is to design a test that incorporates as much information from the covariance structure as possible, without requiring strong assumptions about the covariance operators and while allowing for covariance information to depend on $j$.
\subsection{A likelihood ratio test statistic}
We propose a likelihood ratio-type statistic in the multivariate setting and study its asymptotic behavior as the truncation dimension \( p \to \infty \). Since truncated basis expansions approximate the functions increasingly well as \( p \) grows, no information is lost asymptotically.

The likelihood ratio test statistic is:
\begin{align}
    T_{\mathrm{FLRT}} = \sum_{j=1}^k n_j (\bar{y}_j^{(p)} - \hat{\mu}^{(p)})^\top \hat{\Sigma}_{j,p}^{-1} (\bar{y}_j^{(p)} - \hat{\mu}^{(p)}), \label{finitedim LRT}
\end{align}
where
\[
\hat{\mu}^{(p)} = \left( \sum_{j=1}^k n_j \hat{\Sigma}_{j,p}^{-1} \right)^{-1} \left( \sum_{j=1}^k n_j \hat{\Sigma}_{j,p}^{-1} \bar{y}_j^{(p)} \right)
\]
is the weighted pooled mean estimator. $\bar{y}_j^{(p)}$ and $\hat{\Sigma}_{j,p}$ are defined in (\ref{ypdefinition}) and (\ref{sampletrunccov}) respectively.\\
In fact, $T_{FLRT}$ is nothing more than the classical Gaussian likelihood‐ratio statistic for comparing $k$ multivariate normal populations, with each true covariance matrix $\Sigma_j$ replaced by its sample analogue $\widehat\Sigma_j$. Justified asymptotically by CLT,  standard results from multivariate analysis then yield the following.

\begin{proposition}\label{lem:flrt-chi2}
Let $p$ be fixed, and for each group $j=1,\dots,k$, let
$\{y_{ji}^{(p)}\}_{i=1}^{n_j}$
be independent and identically distributed $p$-variate observations, with the groups independent of each other.  Denote the total sample size by $n=\sum_{j=1}^k n_j$, assume that $\dfrac{n_j}{n} \to c_j>0$ and the existence of second moments.  Then, as $n\to\infty$,
$T_{FLRT}
\;\xrightarrow{d}\;
\chi^2_{p(k-1)}$.
\end{proposition}
The proposition follows directly from the continuous mapping theorem, e.g. (\cite{Dudley2014}).\\
Finally, we assume a multivariate model for the scores:
\begin{align}
    y_{ji}^{(p)} = \Sigma_{j,p}^{1/2} z_{ji} + \mu_j^{(p)}, \label{modelassumption}
\end{align}
where $\Sigma_{j,p}^{1/2}$ is the unique symmetric nonnegative definite square root of $\Sigma_{j,p}$, \( z_{ji} \in \mathbb{R}^p \) has independent entries and $\{z_{ji}\}_{i=1}^{n_j}$ are i.i.d satisfying \( \mathbb{E}[z_{ji}] = 0 \) and \( \operatorname{Var}(z_{ji}) = I_p \), by (\ref{A1}) we see that $\mathbb{E}(z_{jil}^8) < \infty$.

From now on, we will work exclusively with the truncated \( p \)-dimensional representation, so we omit superscripts or subscripts indicating \( p \) for notational simplicity. All vectors and matrices should be understood as being of dimension \( p \), unless otherwise specified.  Define $d_j = \mu_1 - \mu_j$, for all $j = 1, \dots, k$, and
\begin{align}
v &= 
\begin{bmatrix}
    n_1^{1/2} \Sigma_1^{-1/2} d_1 \\
    \vdots \\
    n_k^{1/2} \Sigma_k^{-1/2} d_k
\end{bmatrix}, 
\quad
v_d = 
\begin{bmatrix}
    \Sigma_1^{-1/2} d_1 \\
    \vdots \\
    \Sigma_k^{-1/2} d_k
\end{bmatrix},
\quad
\Sigma =\begin{bmatrix}
    n_1^{1/2} \Sigma_1^{-1/2} \\
    \vdots \\
    n_k^{1/2} \Sigma_k^{-1/2}
\end{bmatrix},
\quad
B = \left( \sum_{j=1}^k n_j \Sigma_j^{-1} \right)^{-1}.
\label{definition of v and sigma}
\end{align}
Note that $P=\Sigma B \Sigma^\top$ is a projection matrix, i.e. $P^2=P^\top=P$.  In (\ref{definition of v and sigma}), $v_d$ determines the local alternative and satisfies
\begin{align}
    ||v_d||^2=o\left (\dfrac{p}{n}\right). \label{localalt}
\end{align}
\begin{theorem} \label{theorem1}
    Consider $T_{FLRT}$, the statistic defined in (\ref{finitedim LRT}), the model for scores defined in (\ref{modelassumption}), $P=\Sigma B \Sigma^\top$ as defined in (\ref{definition of v and sigma}), $Q=I_{pk}-P$ and $n=\sum_{j=1}^k n_j$. Assume that $\dfrac{n_j}{n} \to c_j>0$, $(max_{1 \leq j \leq k} \lambda_{j1}^{-2}) p =o(\sqrt{n})$ and (\ref{localalt}) is satisfied. Then we have the following CLT:
    \begin{align}
        W=\dfrac{T_{FLRT}-p(k-1)-v^\top Qv}{\sqrt{2p(k-1)}} \to N(0,1),
    \end{align}
    where $\lambda_{j1}\le \lambda_{j2}\le \cdots\le \lambda_{jp}$ denote the eigenvalues of $\Sigma_j$ in non-decreasing order.
\end{theorem}
\begin{remark}
    We see from Theorem 1 that the rate of convergence depends on how fast the smallest eigenvalue of the covariance matrix(operator) decreases, this is indeed needed for our test statistic because we kept the inverse sample covariance matrix and for it not to blow up asymptotically one needs such condition, and since we are working with functional data the eigenvalues will go to $0$ at the end. This is not very restrictive condition, due to the fact that in practice, if the functional data has smooth nature, then the smallest eigenvalue goes to zero fast but less scores will be needed, but if the functional data is relatively rough, then the smallest eigenvalue goes to zero at a slower rate, therefore more scores can be used to do the test. The dependence on \(p\) is thus a reasonable dependence because \(p\) also controls how fast the smallest eigenvalue vanishes.
\end{remark}
\begin{remark}
    One sees that $\bar{y}_{j}^{(p)}-\mu_j^{(p)}=O_p(1/\sqrt{n})$ is due to the functional law of large numbers. The dependence on \(p\) is absent in the above due to the nature of functional data. The dependence on $p$ in Theorem 1 is due to whitening and vanishing eigenvalues.
\end{remark}
\begin{corollary}\label{cor:power}
Under the assumptions of Theorem 1, define
\[
s \;=\; \lim_{n,p\to\infty}
\frac{v^\top Q v}{\sqrt{2\,p\,(k-1)}}
\quad.
\]
Then the asymptotic power of the statistic satisfies
\[
\lim_{n,p\to\infty}
\mathbb{P}_{H_a}\!\Bigl(
\dfrac{T_{FLRT}-p(k-1)}{\sqrt{2p(k-1)}} \geq \xi_\alpha
\Bigr)
\;=\;
\begin{cases}
\Phi\bigl(-\xi_\alpha + s\bigr), & s\in\mathbb{R},\\[6pt]
1,                                & s=+\infty,
\end{cases}
\]
where $\xi_\alpha$ is the upper $\alpha$ quantile of $N(0,1)$, i.e. $\xi_\alpha$ is such that $\mathbb{P}[N(0,1)> \xi_\alpha ]=\alpha$.
\end{corollary}

\subsection{Bootstrap}
Bootstrap resampling is a standard device for sharpening finite-sample
approximations.  In fixed dimension its higher-order accuracy is well
documented \citep[e.g.,][]{Hall1992}.  Rigorous results for growing
dimension are scarce, yet in the moderate regime considered here
($p\ll n$) we anticipate that the bootstrap will still outperform the
first-order calibration; Section~\ref{sec:sim} confirms
this empirically.  We therefore adopt a groupwise, nonparametric
bootstrap, using random sampling with replacement, and establish its asymptotic validity for our statistic.

Let $y_{ji}\in\mathbb{R}^{p}$ be the score vector for subject
$i=1,\dots,n_{j}$ in group $j=1,\dots,k$, defined in (\ref{ypdefinition}).  From the empirical
distribution of $\{y_{ji}\}$ we draw, with replacement and within each
group, the bootstrap sample
\begin{equation}
  y_{j1}^{*},\dots,y_{jn_{j}}^{*}\qquad(j=1,\dots,k).
  \label{bootstrap resmaples}
\end{equation}
Write $\bar{y}_{j}^{*}$ and ${\Sigma}_{j}^{*}$ for the
bootstrap mean and covariance, and 
\[
  \widehat{\mu}^{*}
    =\Bigl(\sum_{j=1}^{k} n_{j}\,{\Sigma}_{j}^{*-1}\Bigr)^{-1}
       \Bigl(\sum_{j=1}^{k} n_{j}\,{\Sigma}_{j}^{*-1}
             \bar{y}_{j}^{*}\Bigr).
\]
The bootstrap analogue of the finite-dimensional likelihood-ratio
statistic in~\eqref{finitedim LRT} is
\[
  T_{\text{FLRT}}^{*}
    =\sum_{j=1}^{k}
       n_{j}\bigl(\bar{y}_{j}^{*}-\widehat{\mu}^{*}\bigr)^{\!\top}
         {\Sigma}_{j}^{*-1}
       \bigl(\bar{y}_{j}^{*}-\widehat{\mu}^{*}\bigr).
  \label{finitedim boot LRT}
\]
The following theorem ensures the asymptotic validity of the bootstrap method.
\begin{theorem} \label{theorem2}
    Under the same assumption of theorem 1, we have the following unconditional CLT for the bootstrapped version
    \[  W^*=\dfrac{T^*_{FLRT}-p(k-1)-v^\top Qv}{\sqrt{2p(k-1)}} \to N(0,1).\]
\end{theorem}
The following corollary is a bootstrap analogue of Corollary 1.
\begin{corollary}
    Under the assumptions of Theorem 2, define
\[
s \;=\; \lim_{n,p\to\infty}
\frac{v^\top Q v}{\sqrt{2\,p\,(k-1)}}
\quad.
\]
Then the power of the statistic satisfies
\[
\lim_{n,p\to\infty}
\mathbb{P}_{H_a}\!\Bigl(
\dfrac{T^*_{FLRT}-p(k-1)}{\sqrt{2p(k-1)}} \geq \xi_\alpha
\Bigr)
\;=\;
\begin{cases}
\Phi\bigl(-\xi_\alpha + s\bigr), & s\in\mathbb{R},\\[6pt]
1,                                & s=+\infty,
\end{cases}
\]
where $\xi_\alpha$ is the upper $\alpha$ quantile of $N(0,1)$.
\end{corollary}
The following algorithm can be used as implementation of the bootstrap method. 
\begin{algorithm}
Let $\alpha \in (0,1)$ be the significance level and $B \in \mathbb{N}$ denotes the number of bootstrap samples.
\begin{description}%[label=(\arabic*)]
    \item[\textbf{Step 1}.]Choose basis and dimension $p$; compute the observed scores $y_{j1},...,y_{jn_j},$ \text{ for $1 \leq j \leq k$}.
    \item[\textbf{Step 2}.]  Recentre the observed scores, simulate B times the bootstrap samples using the centred observed scores and get $y_{j1,b}^*,...,y_{jn_j,b}^*,$ \text{for $1 \leq j \leq k$, $1 \leq b \leq B$}.
    \item[\textbf{Step 3}.] Compute the bootstrap test statistic $W^*$ as in theorem (\ref{theorem2}) for all $1 \leq b \leq B$ and the original test statistic $W$.
    \item[\textbf{Step 4}.]  Compute the proportion of the number of times that $W^* \geq W $ for all $1 \leq b \leq B$, denote it by $p_B$.
    \item[\textbf{Step 5}.] Reject the null hypothesis in (\ref{hypothesistest}) if and only if $p_B <\alpha$.
\end{description}
\end{algorithm}
The bootstrap method is basically the same as the fixed dimensional bootstrap method; it is provided here for transparency.

\section{Practical advice for implementation}

In this section, we examine potential limitations of our method and propose strategies to mitigate these potential drawbacks. Although selecting a fixed, deterministic basis set poses no asymptotic issues, practical implementations may encounter challenges. In particular, if the differences between samples are not captured by the leading basis scores—either because the primary discrepancies occur only in later components or they are, approximately, evenly spread across components—the initially chosen basis functions may fail to capture the true differences in the mean functions, if the number of basis functions used is not sufficiently large. We outline several remedies below and describe simulation studies later.\\
From Theorem \ref{theorem1} and Corollary 1 we can see that the noncentrality parameter $v^\top Qv$ plays a primary role in determining the power of our proposed test. We know that $v^\top Qv$ can be computed for all $p \geq 1$ once a basis is chosen. This allows us to detect the performance of the test across the different choices of orthonormal bases and the different numbers of scores used. For present purposes, it would be natural to use the noncentrality parameter $v_d^\top Q v_d$, where $v_d$ is defined in (\ref{definition of v and sigma}).
In practice, however, we are not able to observe the population quantity $v_d^\top Qv_d$, but a sampled version may be used. 
We denote by $\hat{v}_d^\top \hat{Q}\hat{v}_d$ the sample version of $v_d^\top Qv_d$, where the `hat' indicates that we are replacing the population means and covariances by the sample means and covariances, respectively. Classic asymptotic theory shows that
    $\hat{v}_d^\top \hat{Q}\hat{v}_d \xrightarrow{P} v_d^\top Qv_d$ for all $p \geq 1$.
This justifies the use of $\hat{v}_d^\top \hat{Q}\hat{v}_d$ as a practical tool for checking the power under scenarios where $p$ is relatively small compared to $n$.

\subsection{Multiple Bases Diagnostic}
We now propose a type of diagnostic plot for comparing different choices of $L^2$ orthonormal bases and for assessing the choice of $p$.  This plot is generated using the following procedure.
\begin{procedure} Suppose $E=\{e_l\}_{l=1}^\infty$ is an $L^2$ orthonormal basis, $p_{\max}$ is the maximum number of scores considered and $(\hat{v}_d^\top \hat{Q}\hat{v}_d)_p$ denotes $\hat{v}_d^\top \hat{Q}\hat{v}_d$ computed with a particular choice of $p$.  \label{algorithm 2} 
    \begin{description}%[label=(\arabic*)]
    \item[\textbf{Step 1}.] Consider $p=1$, compute the observed scores $y_{j1},...,y_{jn_j},$ \text{ for $1 \leq j \leq k$}.
    \item[\textbf{Step 2}.] Compute the sample means and sample covariance matrices of each population, and then compute $(\hat{v}_d^\top \hat{Q}\hat{v}_d)_1$.
    \item[\textbf{Step 3}.] Repeat the above process for $2 \leq p \leq p_{\max}$, we then observe $\{(\hat{v}_d^\top \hat{Q}\hat{v}_d)_p\}_{p=1}^{p_{max}}$.
    \item[\textbf{Step 4}.] Draw the plot of $(\hat{v}_d^\top \hat{Q}\hat{v}_d)_p$ as a function of $p$.
\end{description}
\end{procedure}
This algorithm generates a plot for a choice of orthonormal basis $E$. If any power exists within the first $p_{\max}$ scores then we would at least expect to see a non-decreasing trend in the plot. If the plot shows a fluctuation below some values that are close to 0 as $p$ varies, then it could be due to one of the following:
(i) the selected basis can not capture the differences between the populations very well using only the first $p_{\max}$ many scores, or (ii) the null hypothesis is true, so there is no difference in means across the populations.

If the plot is indeed increasing, or some big spikes appear, then it indicates that the selected basis is a good candidate for use in the  test statistic and we should have some confidence in the outcome of the test. The sample size affects the volatility of $\hat{v}_d^\top \hat{Q}\hat{v}_d$, therefore one needs more care when looking at the diagnostic plots for small samples.

\subsection{Reordering of Basis Functions Diagnostics}

In this subsection we introduce another strategy for examining what influences power, which is to reorder the basis functions. Since our method is invariant to the ordering, one may order the components by how clearly they distinguish between the different populations.\\
Let $E=\{e_l\}_{l=1}^\infty$ be the chosen $L^2$ orthonormal basis, choose $N \in \mathbb{N}$ to be relatively large and let $p_{c}$ denote the number of scores to be used.   Let $S \subset \{1,...,N\}$ denote a subset of indices to be used, with $|S|=p_{c}$, where $\vert A\vert$ denotes the cardinality of a set $A$. We then define $(\hat{v}_d^\top \hat{Q}\hat{v}_d)_{S}$ to be $\hat{v}_d^\top \hat{Q}\hat{v}_d$ computed with respect to the scores generated by applying $\{e_l\}_{l \in S}$. This allows us to assess the power for different possible combinations of basis functions of the same size. For small $N$ one may compute $(\hat{v}_d^\top \hat{Q}\hat{v}_d)_{S}$ for all possible $S$ and then select the best subset, but this would be computationally heavy if $N$ is large, therefore the following procedure is introduced for implementation and diagnostics:
\begin{procedure} Assume the notation introduced in the preceding paragraph above. \label{algorithm 3} 
\begin{description}%[label=(\arabic*)] 
    \item[\textbf{Step 1}.]  Compute one dimensional scores by using $e_l$ for $l \in \{1,...,N\}$ only, then compute $(\hat{v}_d^\top \hat{Q}\hat{v}_d)_{S=\{l\}}$.
    \item[\textbf{Step 2}.] Repeat the process for all $l \in \{1,...,N\}$, we observe $\{(\hat{v}_d^\top \hat{Q}\hat{v}_d)_{l}\}_{l=1}^N$.
    \item[\textbf{Step 3}.]  Draw the plot of $\{(\hat{v}_d^\top \hat{Q}\hat{v}_d)_{l}\}_{l=1}^N$ versus $i=1, \ldots , N$.
\end{description}    
\end{procedure}
The procedure returns a plot that reflects any differences across samples captured by one-dimensional directions. A point whose magnitude is substantially greater than most others while its absolute numerical value is also large is called spike. This plot can be used for diagnostics in the following way:
\begin{enumerate}[label=(\arabic*)]
    \item If no spikes exist in the plot, no particular choice of subset is needed. The same process can be run for other basis if needed.
    \item If spikes exist, pick the subset of basis functions corresponding to the spikes. The cardinality of the subset should be no more than $p_c$. This subset captures a large amount of information on the differences between the samples. Perform the test with respect to this subset.
\end{enumerate}
In general, one should be  careful if the aim is to reorder a given basis using Procedure 2, due to the fact that the data is being  used twice, which potentially could inflate the size of the test. Techniques such as splitting the dataset into training and testing sets can be applied if using Procedure 2 diagnostic plots to select a basis is needed. Using Procedure 1 diagnostic plots to select a basis may not cause much issue in size inflation. 
Typically, practitioners may perform some preliminary analysis on the dataset to see how well commonly used $L^2$ orthonormal bases (e.g. Fourier and Haar bases) captures the salient features of the data with the largest $p$, $p_{\textrm{max}}$, that we wish to consider, and then pick the one that performs the best.

%Diagnostic plots should often be generated and analysed when fail to reject and lack of confidence. 
%Typically, practitioners may perform some preliminary analysis on the dataset to see how well commonly used $L^2$ orthonormal bases (e.g. Fourier and Haar bases) perform with the largest $p$, $p_{\textrm{max}}$, that we wish to consider, and then pick the one that performs the best.
%in this case the size of the test will not be inflated.

\subsection{The imperfect sampling case}
So far, we have been considering the ideal case where the functional observations are essentially complete. In practice there are situations where the curves are not fully or densely observed and/or measurement error is present.  This case is referred as the imperfect sampling case. A standard set up is as follows:
$Y_{jil}=X_{ji}(t_{l})+\epsilon_{jil}$, where $ 1 \leq j \leq k$, $1 \leq i \leq n_j$ and $1 \leq l \leq m$; the $X_{ji}(t) \in L^2([0,1])$ are random functions; and the $\epsilon_{jil}$ are the measurement errors with mean zero and variance $\sigma^2$, which are independent of each other.\\
Under the sparse set up, $\langle X_{ji},b_q \rangle$ may not be directly observed due to the measurement error and/or small $m$. One could do some pre-smoothing to obtain the scores, which is a common approach in functional data analysis. On the other hand, since the aim is to compute $\langle X_{ji},b_q \rangle$, direct numerical integration is another possible approach. As long as $m$ is not too small, the averaging in numerical integration will make the effect of measurement error small or even negligible, under reasonable conditions. The discrepancy between numerical integration and the true integral may be controlled by both the choice of quadrature method and $m$. Let $\bar{N}_{j,q}$ be the naive estimator for $\langle \mu_{j},b_q \rangle$ using numerical integration. It is shown in Section 8.2 of the Supplemntary Materials that, using the trapezium rule and common sampling design across samples,
\[\bar{N}_{j,q}-\langle \mu_j,b_q\rangle=O_p\left (\frac{1}{\sqrt{n}}\right )+O_p\left (\frac{1}{m^2} \right )+O_p\left (\frac{1}{\sqrt{nm}} \right ),\]
where the first term on the RHS comes from the (functional) law of large numbers; the second term comes from trapezium quadrature rule; and the third term is due to measurement error. Intuitively, even when the curves are sparsely observed, numerical integration  will maintain  accuracy at a suitable  level, under reasonable conditions.

In summary: on the basis of the above comments and Section 8.2 in the Supplementary Materials, we expect that our proposals will give worthwhile results in imperfect sampling settings, under reasonable conditions.
%In summary, we do not lose much from doing numerical integration in the imperfect sampling case.
\section{Simulation results}\label{sec:sim}
In the previous sections, we established the asymptotic validity of the proposed methods. In this section, we assess their finite sample performance through a comprehensive simulation study. 
To gauge finite-sample behaviour we simulate data from Gaussian processes (GPs) on $[0,1]$ whose covariance obeys the Matérn form
\[
k_\nu(x,x')
  \;=\;
  \sigma^{2}\,
  \frac{2^{1-\nu}}{\Gamma(\nu)}
  \Bigl\{\sqrt{2\nu}\,\|x-x'\|/\ell\Bigr\}^\nu
  K_\nu\!\left(\sqrt{2\nu}\,\|x-x'\|/\ell\right),
\]
where $\sigma^{2}$ controls marginal variance, $\ell$ the range, 
$\nu>0$ governs mean-square smoothness and $K_\nu(\cdot)$ is a modified Bessel function of the second kind (see Chapter 9 of \cite{Abramowitz1972}).

Each trajectory is observed on a 100-point grid and generated as a draw from
$\mathcal{N}\!\bigl(\mu,\Sigma(\sigma^{2},\ell,\nu)\bigr)$,
with $\Sigma$ obtained by evaluating $k_\nu$ on the grid.
The curve is then projected onto a fixed orthonormal $L^{2}$ basis to obtain
$p$ scores (numerical integration for dense grids; smoothing first if the grid is sparse).
Varying $(\sigma^{2},\ell,\nu)$ and the mean function produces alternative
populations with either mean or covariance differences.
For every data set we compute the proposed test statistic and calibrate its
null distribution via 1000 bootstrap resamples, using a 5\% nominal level.
Unless otherwise stated we report results from
5\,000 Monte Carlo replications.\\
We comapre our test statistic  with the following tests: (i) the $L^2$-norm-based parametric bootstrap test for heteroscedastic samples due to \cite{Cuevas2004}, referred to as \textbf{CS}; (ii) the $L^{2}$-norm-based bootstrap test due to \cite{Zhang2014}, referred to as \textbf{L2b}; (iii) the $F$-type test due to Zhang (2014), referred to as \textbf{Fb}; (iv) the F\(_\text{max}\) bootstrap test of \cite{ZhangChengWuZhou2019}, which takes the pointwise maximum of  \(F\) statistics, referred to as \textbf{Fmaxb}; and (v) the globalising-pointwise-\(F\) test of \cite{Zhang2014a},
      which combines pointwise \(F\) evidence into a single omnibus
      ratio, referred to as \textbf{GPF}.  We refer to \textbf{TLRT} as our proposed test statistic. In the mean time, at all cases for $k=2$, \textbf{T2} is referred to as the usual multivariate Hotelling $T^2$ statistic.
%\begin{itemize}[leftmargin=1.5em,itemsep=2pt]
%\item \textbf{CS} $L2$-norm-based parametric bootstrap tests heteroscedastic samples \cite{Cuevas2004}.
%\item \textbf{L2b} – an $L^{2}$-norm-based bootstrap test by \cite{Zhang2014}.
%\item \textbf{Fb} – the F type test by \cite{Zhang2014}.
%\item \textbf{Fmaxb} –  The F\(_\text{max}\) bootstrap test of \cite{ZhangChengWuZhou2019}, which
%      takes the maximum of pointwise \(F\) statistics. 
%\item \textbf{GPF} – The globalising-pointwise-\(F\) test of \cite{Zhang2014a},
     % which combines pointwise \(F\) evidence into a single omnibus
      %ratio. 
%\end{itemize}

Table \ref{tab:combined_single} gives a study of test size with two populations ($k=2$) with sample sizes $n=(50,30)$.  
Curves are sampled from GPs with Matérn covariance; population~1 uses $(\ell,\sigma^{2})=(1,5)$, population~2 uses $(4,1)$.  
The smoothness parameter $\nu$ is varied to probe its effect on type-I error.  
For each simulated data set we compute TLRT (Haar/Fourier), Hotelling (Haar/Fourier), CS, L2b, Fb, Fmaxb, GPF and Horváth(the fixed dimensional projection method of \cite{Horvath2015}).  
Fmaxb, GPF, and Horváth’s test do not perform well and are therefore omitted from subsequent power analyses. The performance of Horváth's text can be seen in Table \ref{tab:table3} in the Supplementary Materials.\\
Table \ref{tab:table5} gives details of power versus mean-function shape.
The same design ($k=2$, $n=(50,30)$) is retained, but now both GPs share $(\ell,\sigma^{2},\nu)=(1,5,5)$ for Population~1 and $(0.5,1,5)$ for Population~2.  
Mean functions are $\mu_{1}(x)=cx$ and $\mu_{2}(x)=-cx^{2}$ with $0\le c\le1.39$ (the $L^{2}$ distance between $x$ and $-x^{2}$ is~$\approx1$).  
With only two basis functions ($p=2$) the competing methods perform similarly (see Table \ref{tab:table4} in the Supplementary Materials); adding a third score ($p=3$) captures the curvature difference and pushes our test ahead of the alternatives.\\
\begin{table}[H]
    \centering
    \footnotesize
    \caption{Actual size of different tests with varying \(\nu\).  The nominal size is 0.05.}
    \label{tab:combined_single}
    \begin{tabular}{c c c c c c c c c c}
        \toprule
        \(\nu\) & TLRT HB & TLRT FB & T2 HB & T2 FB 
        & CS & L2b & Fb & Fmaxb & GPF \\
        \midrule
        0.5  & 0.0520 & 0.0528 & 0.0558 & 0.0516 & 0.0542 & 0.0550 & 0.0484 & 0.0160 & 0.0274 \\
        1.0  & 0.0480 & 0.0488 & 0.0436 & 0.0486 & 0.0494 & 0.0484 & 0.0420 & 0.0140 & 0.0214 \\
        1.5  & 0.0500 & 0.0504 & 0.0500 & 0.0490 & 0.0554 & 0.0556 & 0.0490 & 0.0182 & 0.0264 \\
        2.0  & 0.0464 & 0.0434 & 0.0456 & 0.0430 & 0.0520 & 0.0524 & 0.0460 & 0.0150 & 0.0216 \\
        5.0  & 0.0462 & 0.0430 & 0.0450 & 0.0448 & 0.0506 & 0.0504 & 0.0446 & 0.0164 & 0.0216 \\
        10.0 & 0.0468 & 0.0456 & 0.0484 & 0.0452 & 0.0542 & 0.0554 & 0.0480 & 0.0168 & 0.0212 \\
        50.0 & 0.0486 & 0.0480 & 0.0476 & 0.0478 & 0.0554 & 0.0568 & 0.0478 & 0.0176 & 0.0244 \\
        \bottomrule
    \end{tabular}
\end{table}
Table \ref{tab:hf-power} considers a high-frequency shift.
Larger samples $n=(500,300)$ are generated from GPs with $(\ell,\sigma^{2},\nu)=(1,5,5)$ and $(0.5,1,5)$.  
The means are $\mu_{1}(x)=x$ and $\mu_{2}(x)=x+c\sqrt{2}\sin(10\pi x)$, $0\le c\le0.6$, so the signal resides in the 11th Fourier coefficient.  
Consequently the Fourier expansion gains power only once $p\ge11$, whereas the Haar basis detects the discrepancy by $p=4$, illustrating the benefit of a flexible basis choice.\\
Figure 1 shows the multiple bases diagnostic plots for Haar and Fourier bases under the same setting as Table \ref{tab:hf-power}. We can see that there is a significant jump for the Haar basis at $p=4$, which matches with what the simulation gives. The Fourier basis starts to have a significant jump at $p=11$, which matches with the simulation.\\
Figure 2 shows the plots formed by Procedure \ref{algorithm 3} under the same setting as Table \ref{tab:hf-power} for unequal mean and an equal mean case with both population means being $x$, where $N=100$ was applied for both cases. For the equal mean plots, we see that for both Haar and Fourier bases, the numerical values of all points are very close to zero, which agrees with our expectation under the null hypothesis. For the unequal mean plots, the Fourier basis shows a significant spike at index $11$, which corresponds with the simulations and the population set-up; the Haar basis shows many points with large absolute numerical value at early basis functions, which also supports the simulation result in Table \ref{tab:hf-power}. 
\begin{table}[H]
    \centering
    \footnotesize
    \begin{tabular}{c c c c c c c c}
        \toprule
        c Value & TLRT Haar & TLRT Fourier & T2 Haar & T2 Fourier & CS & L2b & Fb \\
        \midrule
         0.000 & 0.4580 & 0.0434 & 0.0444 & 0.0434 & 0.0550 & 0.0550 & 0.0484 \\
         0.107 & 0.0736 & 0.0718 & 0.0708 & 0.0684 & 0.0652 & 0.0666 & 0.0566 \\
         0.214 & 0.1560 & 0.1360 & 0.1420 & 0.1250 & 0.0920 & 0.0914 & 0.0806 \\
         0.321 & 0.3380 & 0.2870 & 0.3080 & 0.2650 & 0.1540 & 0.1520 & 0.1370 \\
         0.429 & 0.5730 & 0.4890 & 0.5340 & 0.4670 & 0.2290 & 0.2290 & 0.2120 \\
         0.536 & 0.7780 & 0.6920 & 0.7470 & 0.6620 & 0.3410 & 0.3420 & 0.3170 \\
         0.643 & 0.9260 & 0.8670 & 0.9090 & 0.8440 & 0.5020 & 0.5040 & 0.4700 \\
         0.750 & 0.9770 & 0.9480 & 0.9660 & 0.9350 & 0.6440 & 0.6470 & 0.6160 \\
         0.857 & 0.9960 & 0.9860 & 0.9930 & 0.9830 & 0.7850 & 0.7840 & 0.7560 \\
         0.964 & 1.0000 & 0.9970 & 1.0000 & 0.9950 & 0.8920 & 0.8910 & 0.8740 \\
         1.070 & 1.0000 & 0.9990 & 1.0000 & 1.0000 & 0.9550 & 0.9560 & 0.9460 \\
         1.180 & 1.0000 & 1.0000 & 1.0000 & 1.0000 & 0.9830 & 0.9830 & 0.9780 \\
         1.290 & 1.0000 & 1.0000 & 1.0000 & 1.0000 & 0.9950 & 0.9950 & 0.9930 \\
        \bottomrule
    \end{tabular}
    \caption{Proportion of rejections for different methods at \(p = 3\) across various \(c\) values.}
    \label{tab:table5}
\end{table}

These two figures show that, even when the major difference between populations comes from a specific basis function (in our simulation it is the $11th$ Fourier basis function), changing the basis from Fourier to Haar can help us to detect the difference in earlier components of Haar basis functions. Score(s) that identify differences can be detected by diagnostic plots using Procedure \ref{algorithm 3}.

\begin{table}[t]
\centering
\footnotesize
\caption{TLRT rejection rates ($\alpha=0.05$) under the high–frequency
         mean‐shift design.}
\label{tab:hf-power}

\begin{minipage}{0.30\textwidth}
\centering
\textbf{$p=2,3$}\\[4pt]
\begin{tabular}{cccc}
\toprule
$p$ & $c$ & Haar & Fourier \\ \midrule
\multirow{4}{*}{2}
 & 0.0 & 0.0430 & 0.0436 \\
 & 0.2 & 0.0552 & 0.0506 \\
 & 0.4 & 0.0586 & 0.0504 \\
 & 0.6 & 0.0732 & 0.0568 \\ \midrule
\multirow{4}{*}{3}
 & 0.0 & 0.0470 & 0.0478 \\
 & 0.2 & 0.0528 & 0.0522 \\
 & 0.4 & 0.0550 & 0.0414 \\
 & 0.6 & 0.0686 & 0.0470 \\
\bottomrule
\end{tabular}
\end{minipage}
\hfill
% ---------- p = 4,5 ----------
\begin{minipage}{0.30\textwidth}
\centering
\textbf{$p=4,5$}\\[4pt]
\begin{tabular}{cccc}
\toprule
$p$ & $c$ & Haar & Fourier \\ \midrule
\multirow{4}{*}{4}
 & 0.0 & 0.0522 & 0.0548 \\
 & 0.2 & 1.0000 & 0.0530 \\
 & 0.4 & 1.0000 & 0.0494 \\
 & 0.6 & 1.0000 & 0.0498 \\ \midrule
\multirow{4}{*}{5}
 & 0.0 & 0.0542 & 0.0544 \\
 & 0.2 & 1.0000 & 0.0494 \\
 & 0.4 & 1.0000 & 0.0460 \\
 & 0.6 & 1.0000 & 0.0542 \\
\bottomrule
\end{tabular}
\end{minipage}
\hfill
% ---------- p = 10,11 ----------
\begin{minipage}{0.30\textwidth}
\centering
\textbf{$p=10,11$}\\[4pt]
\begin{tabular}{cccc}
\toprule
$p$ & $c$ & Haar & Fourier \\ \midrule
\multirow{4}{*}{10}
 & 0.0 & 0.0492 & 0.0484 \\
 & 0.2 & 1.0000 & 0.0502 \\
 & 0.4 & 1.0000 & 0.0472 \\
 & 0.6 & 1.0000 & 0.0466 \\ \midrule
\multirow{4}{*}{11}
 & 0.0 & 0.0586 & 0.0560 \\
 & 0.2 & 1.0000 & 1.0000 \\
 & 0.4 & 1.0000 & 1.0000 \\
 & 0.6 & 1.0000 & 1.0000 \\
\bottomrule
\end{tabular}
\end{minipage}
\end{table}

\begin{figure}[H]          % H = exactly "here"
  \centering
  \includegraphics[width=0.5\textwidth]{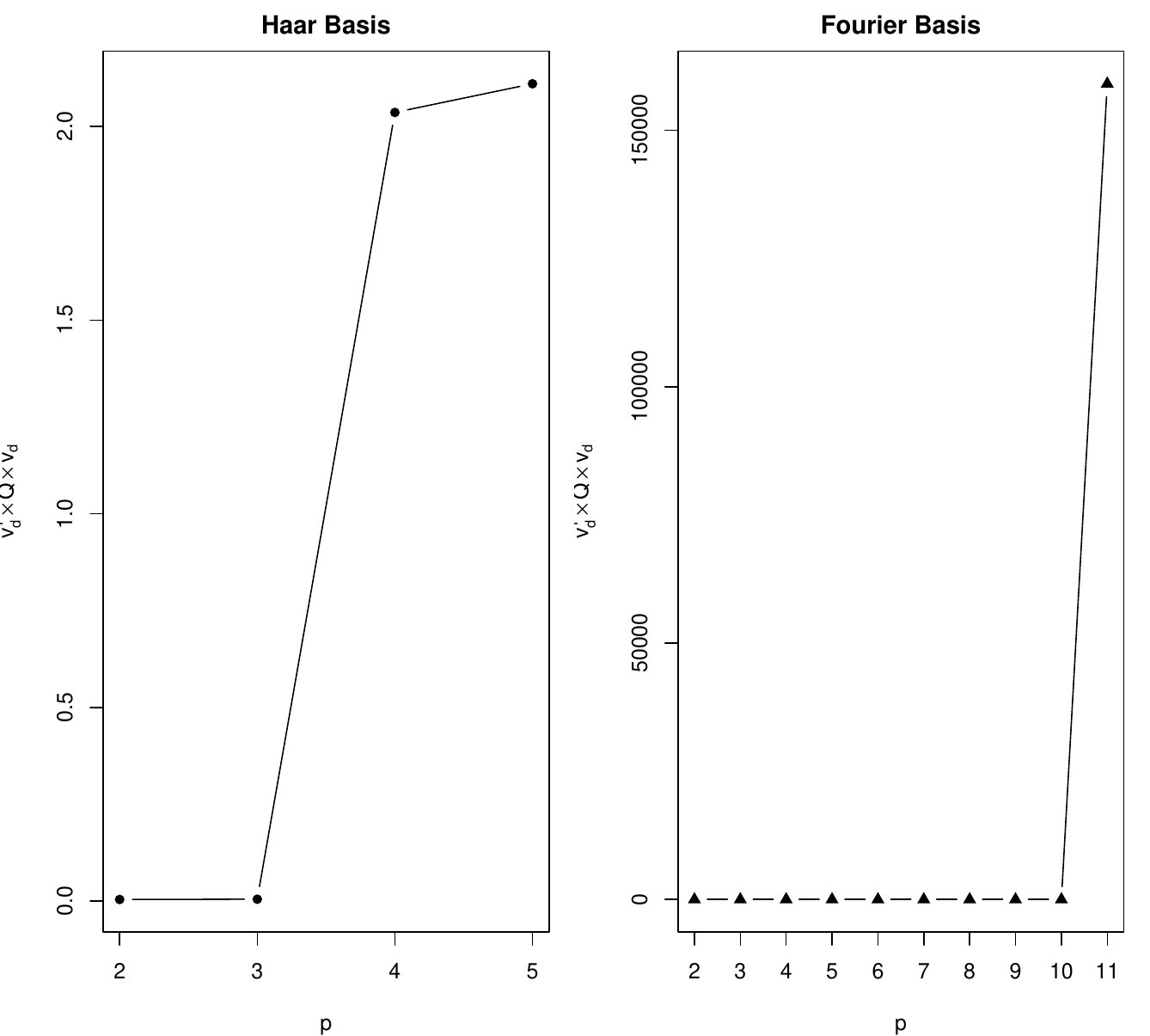}
  \caption{Diagnostic plots for Procedure 1 comparing Haar and Fourier bases.}
  \label{fig:haar-fourier-diagnostic}
\end{figure}
\begin{figure}[H]          % H = exactly "here"
  \centering
  \includegraphics[height=0.5\textwidth]{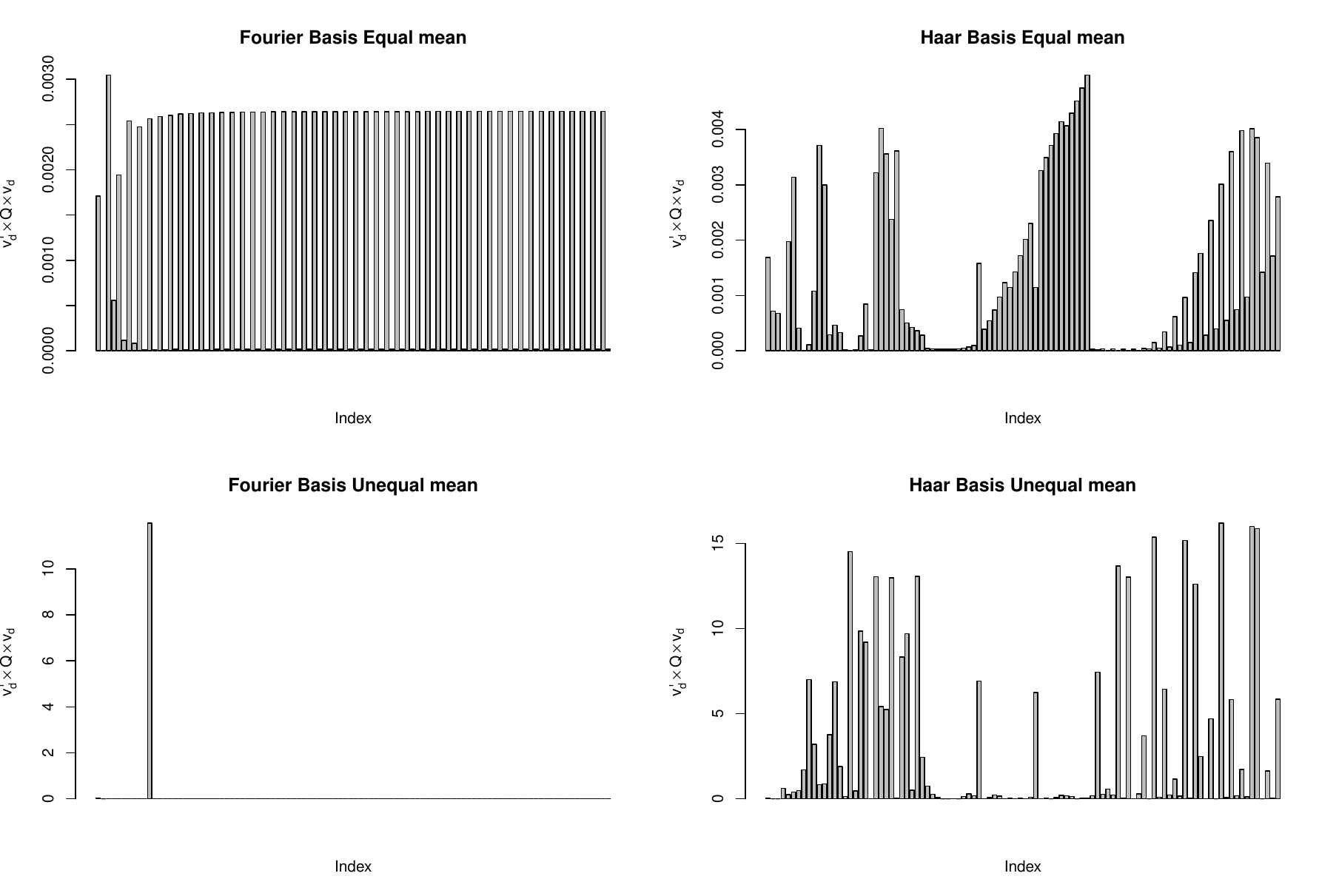}
  \caption{Diagnostic plots for Procedure 2 under different settings.}
  \label{fig:reordering-diagnostic}
\end{figure}

\section{Applications to two real world datasets}
In this section, we analyse two real world data examples using our proposed methods.
\subsection{ECG Data}
We use the \textbf{ECG5000} benchmark from the UCR/UEA Time–Series Classification Archive, where it was originally published in \cite{Goldberger2000PhysioNet}, and underwent pre-processing in \cite{Chen2015NEL}. The series are derived from record \texttt{chf07} of the BIDMC Congestive Heart Failure Database on PhysioNet, an approximately 20-hour, two-lead ambulatory ECG sampled at 250\,Hz from a 48-year-old man (NYHA class III–IV). Each observation is a single heartbeat segment that has been extracted and then interpolated to a common length of 140 samples; the time index is phase-normalized within the R--R interval, and class labels are obtained by automated annotation. Following the archive convention, we adopt the official split of \(n=500\) training and \(n=4{,}500\) test instances. Beats are categorized into five clinically meaningful rhythm types: normal sinus (N), R-on-T premature ventricular contraction (R-on-T PVC), premature ventricular contraction (PVC), supraventricular premature/ectopic beat (SP/EB), and unclassified beat (UB). Table~\ref{tab:ecg5000} reports the class frequencies for the training and test folds. Figure~\ref{ECG} shows representative training traces and classwise mean functions (UB omitted due to its small sample size).
\begin{table}[htbp]
\centering
\footnotesize
\caption{Class distribution in the official UCR/UEA split of the
ECG5000 dataset.}
\label{tab:ecg5000}
\begin{tabular}{lrrr}
\toprule
Rhythm type & Training ($n=500$) & Test ($n=4\,500$)\\
\midrule
Normal (N)                                   & 292 & 2\,627 \\
R-on-T premature ventricular contraction      & 177 & 1\,590 \\
Premature ventricular contraction (PVC)       &  10 &    86 \\
Supra-ventricular premature / ectopic beat    &  19 &   175 \\
Unclassified beat (UB)                        &   2 &    22 \\
\bottomrule
\end{tabular}
\end{table}
\begin{figure}[H]          % H = exactly "here"
  \centering
  \includegraphics[height=0.5\textwidth]{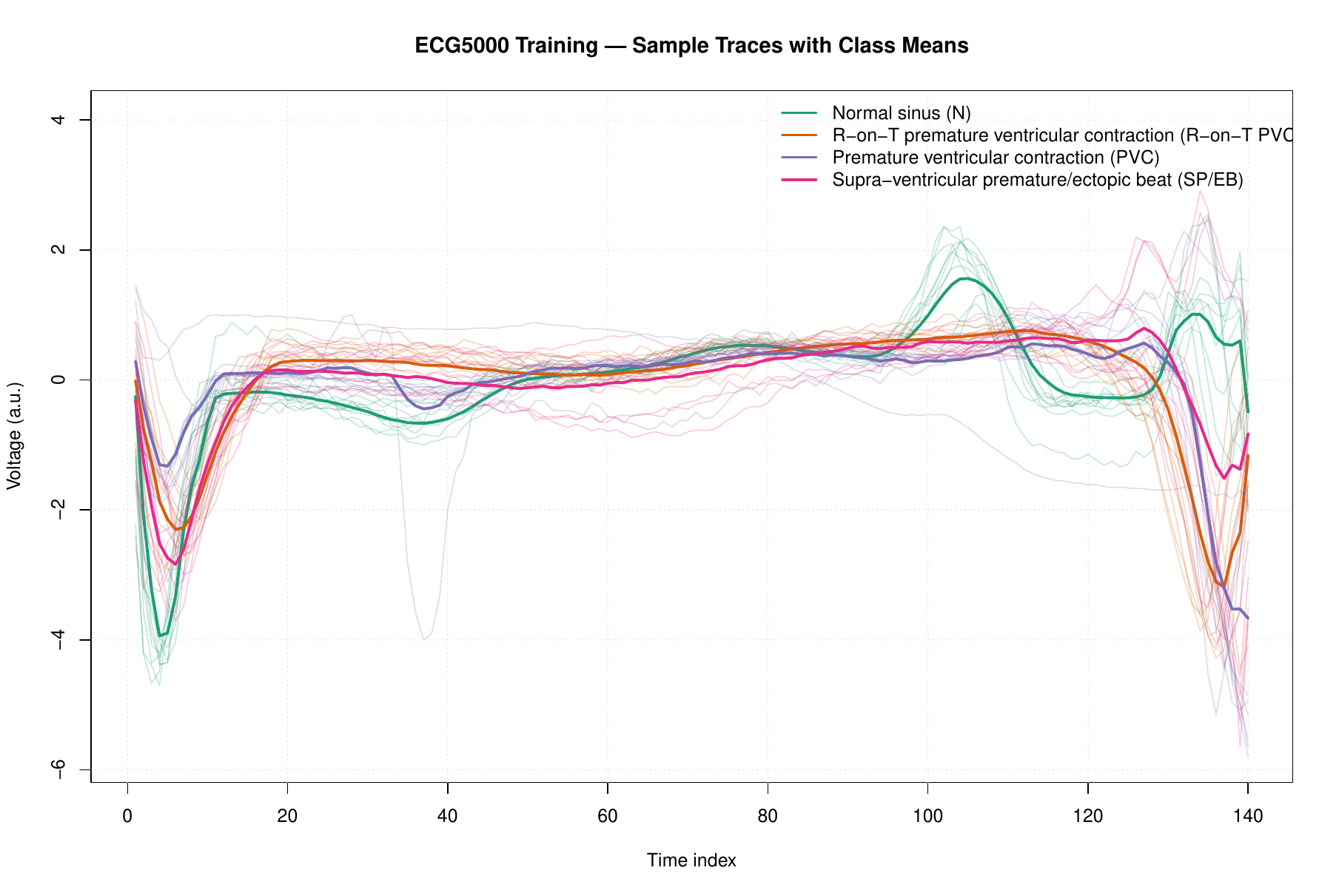}
  \caption{Several training traces per ECG5000 class (thin) with the
classwise mean curve (bold) over 140 time points for the first four classes.}
  \label{ECG}
\end{figure}

Methodologically, \texttt{ECG5000} is well suited to our framework: after phase normalization, each beat is a square-integrable curve on \([0,1]\) whose arrhythmic variations appear as localized shape changes within the P–QRS–T complex. We expand the curves in a common orthonormal basis and use the resulting projection scores to implement our proposed test of equality of mean functions across rhythm categories.

Two different hypothesis tests have been performed with various choices of the number of scores $p$. 
For Table \ref{tab:tlrt-split} (see the Supplementary Materials) we see that Haar rejects starting at $p=2$, where Fourier does not reject at $p=2$ and rejects for the rest.
For Table \ref{tab:tlrt-split-run2} we see that Haar rejects at $p=2$, and Fourier starts to reject at $p=6$. The sample size for the hypothesis in Table \ref{tab:tlrt-split-run2}, $n=(10,19)$, is quite small.  Therefore there could be a large discrepancy between the bootstrap p-value and the asymptotic p-value.

\begin{table}[H]
  \centering
  \footnotesize
  \caption{Summary under $H_{0}\!: \mu_1=\mu_2=\mu_{3}=\mu_{4}$. For $p=5,6,7,8,9$, all $p$-values were zero to three decimal places.}
  \label{tab:tlrt-split}

  \begin{tabular}{%
      c
      >{\raggedleft\arraybackslash}p{2.9cm}
      >{\raggedleft\arraybackslash}p{1.6cm}
      @{\hspace{2em}}
      c
      >{\raggedleft\arraybackslash}p{2.9cm}
      >{\raggedleft\arraybackslash}p{1.6cm}}
    \toprule
      & \multicolumn{2}{c}{\textbf{Fourier basis}}
      & \multicolumn{2}{c}{\textbf{Haar basis}} \\[2pt]
      \cmidrule(lr){2-3} \cmidrule(lr){4-6}
      \(p\) & Test statistic & Bootstrap \(p\)-value
      & \(p\) & Test statistic & Bootstrap \(p\)-value \\
    \midrule
       2 & \makecell[r]{29.3273\\ {\scriptsize p-value = 0}} & 0.067
         &  2 & \makecell[r]{81.5302\\ {\scriptsize p-value = 0}} & 0.000 \\
       3 & \makecell[r]{53.6410\\ {\scriptsize p-value = 0}} & 0.036
         &  3 & \makecell[r]{101.6245\\ {\scriptsize p-value = 0}} & 0.000 \\
       4 & \makecell[r]{191.6895\\ {\scriptsize p-value = 0}} & 0.001
         &  4 & \makecell[r]{187.3650\\ {\scriptsize p-value = 0}} & 0.000 \\
      % 5 & \makecell[r]{216.1331\\ {\scriptsize p-value = 0}} & 0.000
      %   &  5 & \makecell[r]{326.0091\\ {\scriptsize p-value = 0}} & 0.000 \\
      % 6 & \makecell[r]{635.7618\\ {\scriptsize p-value = 0}} & 0.000
      %   &  6 & \makecell[r]{350.3186\\ {\scriptsize p-value = 0}} & 0.000 \\
      % 7 & \makecell[r]{700.9740\\ {\scriptsize p-value = 0}} & 0.000
      %   &  7 & \makecell[r]{700.8406\\ {\scriptsize p-value = 0}} & 0.000 \\
      % 8 & \makecell[r]{755.2110\\ {\scriptsize p-value = 0}} & 0.000
      %   &  8 & \makecell[r]{808.7973\\ {\scriptsize p-value = 0}} & 0.000 \\
      % 9 & \makecell[r]{899.0088\\ {\scriptsize p-value = 0}} & 0.000
      %   &  9 & \makecell[r]{817.1826\\ {\scriptsize p-value = 0}} & 0.000 \\
    \bottomrule
  \end{tabular}
\end{table}

\begin{table}[H]
  \centering
  \footnotesize
  \caption{Summary under the null hypothesis $H_{0}\!: \mu_{3}=\mu_{4}$.}
  \label{tab:tlrt-split-run2}

  \begin{tabular}{%
      c
      >{\raggedleft\arraybackslash}p{2.9cm}
      >{\raggedleft\arraybackslash}p{1.6cm}
      @{\hspace{2em}}
      c
      >{\raggedleft\arraybackslash}p{2.9cm}
      >{\raggedleft\arraybackslash}p{1.6cm}}
    \toprule
      & \multicolumn{2}{c}{\textbf{Fourier basis}}
      & \multicolumn{2}{c}{\textbf{Haar basis}} \\[2pt]
      \cmidrule(lr){2-3} \cmidrule(lr){4-6}
      \(p\) & Test statistic & Bootstrap \(p\)-value
      & \(p\) & Test statistic & Bootstrap \(p\)-value \\
    \midrule
       2 & \makecell[r]{2.3134\\ {\scriptsize p-value = 0.0208}} & 0.080
         &  2 & \makecell[r]{15.5921\\ {\scriptsize p-value = 0}} & 0.001 \\
       3 & \makecell[r]{2.8099\\ {\scriptsize p-value = 0.0050}} & 0.317
         &  3 & \makecell[r]{15.1621\\ {\scriptsize p-value = 0}} & 0.002 \\
       4 & \makecell[r]{8.4024\\ {\scriptsize p-value = 0}} & 0.146
         &  4 & \makecell[r]{51.2018\\ {\scriptsize p-value = 0}} & 0.000 \\
       5 & \makecell[r]{8.0698\\ {\scriptsize p-value = 0}} & 0.208
         &  5 & \makecell[r]{54.1445\\ {\scriptsize p-value = 0}} & 0.001 \\
       6 & \makecell[r]{14.3369\\ {\scriptsize p-value = 0}} & 0.046
         &  6 & \makecell[r]{50.2278\\ {\scriptsize p-value = 0}} & 0.001 \\
       7 & \makecell[r]{16.9438\\ {\scriptsize p-value = 0}} & 0.016
         &  7 & \makecell[r]{53.2479\\ {\scriptsize p-value = 0}} & 0.000 \\
       8 & \makecell[r]{18.6811\\ {\scriptsize p-value = 0}} & 0.010
         &  8 & \makecell[r]{52.0832\\ {\scriptsize p-value = 0}} & 0.000 \\
       9 & \makecell[r]{34.9897\\ {\scriptsize p-value = 0}} & 0.002
         &  9 & \makecell[r]{51.8583\\ {\scriptsize p-value = 0}} & 0.000 \\
    \bottomrule
  \end{tabular}
\end{table}

\subsection{Growth Data}
We complement the ECG experiment with the well-known \textbf{Berkeley Growth Study} heights available in \texttt{fda::growth}. The study followed $n=93$, $39$ boys and $54$ girls—from ages $1$ to $18$ years. Stature (cm) was recorded at a common set of $M=31$ ages that are \emph{not} equally spaced, with denser sampling during periods of rapid growth. The object provides the matrices \verb|hgtm| ($31\times39$) and \verb|hgtf| ($31\times54$), containing discretised height curves for boys and girls, respectively, and the vector \verb|age| (length $31$) of sampling times (years). Figure \ref{growth} shows the classwise gender mean curves and several traces.
\begin{figure}[H]          % H = exactly "here"
  \centering
  \includegraphics[height=0.5\textwidth]{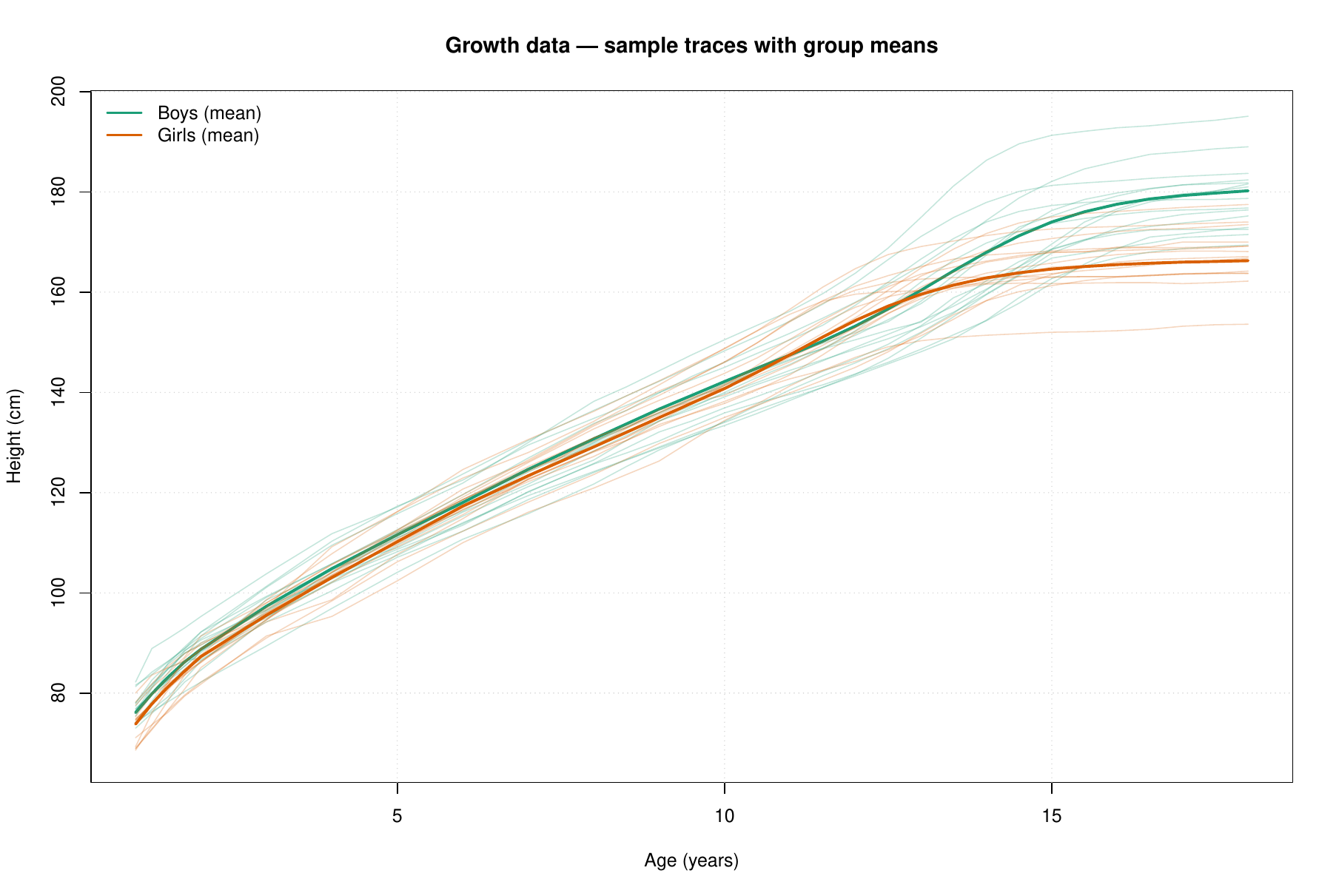}
  \caption{Several traces per gender class (thin) with the
classwise mean curve (bold) over time.}
  \label{growth}
\end{figure}

Each column of \verb|hgtm| or \verb|hgtf| is an essentially monotone growth trajectory exhibiting the characteristic pubertal spurt in mid-adolescence. We obtain functional representations by smoothing these trajectories with B-spline bases (knots at the observed ages), after which the proposed inference procedures test equality of mean growth functions between sexes.

%\begin{table}[htbp]
%\centering
%\caption{Summary of the Berkeley Growth Study sample.}
%\label{tab:growth-summary}
%\begin{tabular}{lrrrr}
%\toprule
%Sex & Number of children & Age points ($M$) & Age range (years) & Units \\
%\midrule
%Boys  & 39 & 31 & $1\text{–}18$ & cm \\
%Girls & 54 & 31 & $1\text{–}18$ & cm \\
%\bottomrule
%\end{tabular}
%\end{table}

For the Berkeley Growth Study height curves, we tested the null hypothesis 
\(H_{0}:\mu_{\text{boys}} = \mu_{\text{girls}}\) using four scores with a B-spline basis. 
The observed test statistic was \(116.72\) on $4$ degrees of freedom, and the bootstrapped \(p\)-value 
was \(0\), providing strong evidence against the null hypothesis of equal mean height functions.

\section{Outlines of proofs of Theorems}
In this section we present proofs of  Theorem 1 and Theorem 2 stated above, based on Lemmas 1, 2 and 3 stated below. These lemmas are proved in the Supplementary Materials.\\
Consider the hypothesis testing framework in (\ref{hypothesistest}), let us first assume that we are working with population covariance matrices $\Sigma_1, \ldots , \Sigma_k$, and we define 
\begin{align}
    T_{\mathrm{FLRT}}^{\mathcal{P}} = \sum_{j=1}^k n_j (\bar{y}_j - \hat{\mu}_{\mathcal{P}})^\top \Sigma_{j}^{-1} (\bar{y}_j - \hat{\mu}_{{\mathcal{P}}}), \hskip 0.2truein
    \hat{\mu}_{\mathcal{P}} = \left( \sum_{j=1}^k n_j \Sigma_{j}^{-1} \right)^{-1} \left( \sum_{j=1}^k n_j \Sigma_{j}^{-1} \bar{y}_j \right), \label{population test stat}
\end{align}
to be the population version of our test statistic.\\
We define the following:
\begin{equation}
    \bar{z}_j = \Sigma_j^{-1/2}(\bar{y}_j - \mu) = \frac{1}{n_j} \sum_{i=1}^{n_j} \Sigma_j^{-1/2}(y_{ji} - \mu) = \frac{1}{n_j} \sum_{i=1}^{n_j} z_{ji}, \hskip 0.2truein
    \bar{z} = \left [
\frac{1}{\sqrt{n_1}} \sum_{i=1}^{n_1} z_{1i}, \ldots , 
\frac{1}{\sqrt{n_k}} \sum_{i=1}^{n_k} z_{ki} \right ]^\top,
    \label{standardised sample mean}
\end{equation}
where the vectors \( z_{ji} \) are independent across both \( i \) and \( j \).\\
Define an embedding operator \( T_j: \mathbb{R}^p \to \mathbb{R}^{pk} \) as
$T_j(w) = [0_{p(j-1)}^\top, w^\top, 0_{p(k-j)}^\top]^\top$,
where \( w \) occupies the \( j \)th $p \times 1$ block and all other $p \times 1$ blocks consist of zeros.
Now, for \( m \in \{1, \dots, n\} \), define
\begin{align}
    x_m = T_j\left( \frac{1}{\sqrt{n_j}} z_{ji} \right)
\quad \text{if} \quad m = i + \sum_{d=1}^{j-1} n_d, \label{x_is}
\end{align}
where \( j \in \{1, \dots, k\} \) and \( i \in \{1, \dots, n_j\} \). Since this mapping is unique, we denote the unique \( j \) and \( i \) corresponding to \( m \) by \( j(m) \) and \( i(m) \), respectively.

The full test statistic can be decomposed as
$T^{\mathcal{P}}_{\mathrm{FLRT}} = M + 2E_1 + E_2$,
where
\begin{align}
    M = \left( \sum_{s=1}^n x_s \right)^\top Q \left( \sum_{s=1}^n x_s \right), \quad
E_1 = v^\top Q \left( \sum_{s=1}^n x_s \right), \quad
E_2 = v^\top Q v, \label{M E1 E2},
\end{align}
where $v$ is defined in (\ref{definition of v and sigma}) and $Q$ is defined in Theorem 1; also, $\sum_{s=1}^n x_s=\bar{z}$, defined in (\ref{standardised sample mean}).\\
Let \( J_{pk,j} \) be the block diagonal matrix in \( \mathbb{R}^{pk \times pk} \) with the \( j \)th block equal to \( I_p \), with all other blocks consisting of zeros. Then \( I_{pk} = \sum_{j=1}^k J_{pk,j} \). Define 
\begin{align}
    I_{pk,m} = \frac{1}{n_{j(m)}} J_{pk, j(m)}. \label{definition of Ipkm}
\end{align}
Using the properties \( \mathbb{E}[z_{ji}] = 0 \), \( \mathrm{Var}(z_{ji}) = I_p \), and independence across \( i, j \), we have $\mathbb{E}[x_m] = 0, \, \mathrm{Var}(x_m) = I_{pk,m},$
and \( x_i \) and \( x_j \) are independent for \( i \ne j \).

We consider \( Q \) as a \( k \times k \) block matrix with each block \( Q_{i,j} \in \mathbb{R}^{p \times p} \). We define
\[
S_n = T_{FLRT}^{\mathcal{P}}- p(k - 1) - v^\top Q v
= \left( \sum_{s=1}^n x_s \right)^\top Q \left( \sum_{s=1}^n x_s \right)+ 2v^\top Q \left( \sum_{s=1}^n x_s \right)
- p(k - 1).
\]
Define also the filtration \( \mathcal{F}_{nm} = \sigma(x_1, \dots, x_m) \), $m \geq 1$, and construct the martingale sequence \( M_{n,m} = \mathbb{E}[S_n | \mathcal{F}_{nm}] \). Let
\begin{align}
    K_m = \sum_{s=1}^m x_s, \quad D_{nm} = M_{n,m} - M_{n,m-1}, \quad \sigma_n^2 = 2(tr(Q)+2v^\top Qv). \label{definition of KM DNM SIGMAN^2}
\end{align}
By the tower property of conditional expectation and independence, we compute:
\[
\mathbb{E}[x_s^\top Q x_t] = 0 \quad \text{for } s \ne t, \quad
\mathbb{E}[x_s^\top Q x_s] = \mathrm{tr}(I_{pk,s} Q) = \frac{1}{n_{j(s)}} \mathrm{tr}(Q_{j(s),j(s)}).
\]

Next, we compute \( M_{n,m} \) explicitly:
\begin{align*}
M_{n,m} &= K_m^\top Q K_m + \mathrm{tr}\left( \sum_{s = m+1}^n I_{pk,s} Q \right) + 2v^\top Q K_m - p(k - 1).
\end{align*}

Therefore,
\begin{align}
D_{nm} &=M_{n,m}-M_{n,m-1}= 2x_m^\top Q K_{m-1} + x_m^\top Q x_m - \mathrm{tr}(I_{pk,m} Q) + 2v^\top Q x_m. \label{Dnm}
\end{align}
\subsection{Main Lemmas}
We state the following three lemmas before we prove Theorem 1.  The proofs of these three lemmas can be found in the Supplementary Materials.
\begin{lemma} Under the local alternative (\ref{localalt}), \label{lemma2}
    $\sum_{m=1}^n \mathbb{E}[D_{nm}^2 | \mathcal{F}_{nm-1}]/\sigma_n^2   \rightarrow 1$ in probability,
    where all terms in this expression   are defined in or immediately above (\ref{definition of KM DNM SIGMAN^2}).
\end{lemma}
\begin{lemma} Under the local alternative (\ref{localalt}), \label{lemma3}
     $\sum_{m=1}^n \mathbb{E}( D^2_{nm} I_{(|D_{nm}| \geq \sigma\epsilon)})/\sigma_n^2 \to 0$,
    where all terms in this expression are defined in or immediately above (\ref{definition of KM DNM SIGMAN^2}).
\end{lemma}
\begin{lemma} \label{lemma4}
    Let $\lambda_{j1} \leq ... \leq \lambda_{jp}$ be the eigenvalues of $\Sigma_j$. Consider $p \leq n$, $\max_{1 \leq j \leq k} \lambda_{j1}^{-1}\sqrt{p}=o(\sqrt{n})$. Then
    \[||P-\hat{P}||_2 = O_p \left (\dfrac{(max_{1 \leq j \leq k}\lambda_{j1}^{-2})\sqrt{p}}{\sqrt{n}} \right),\]
    where $\hat{P}$ is the sampled version of $P$ by replacing the population covariance matrices with the sample covariance matrices defined in Theorem 1.
\end{lemma}
\begin{proof}[Proof of Theorem \ref{theorem1}]
By Lemma \ref{lemma4}, we directly obtain the bound
\[
\frac{\left| \bar{z}^\top (\hat{P} - P) \bar{z} \right|}{\sqrt{2p(k - 1)}} = O_p\left( \frac{ \max_{1 \leq j \leq k} \lambda_{j1}^{-2} \cdot p }{ \sqrt{n} } \right).
\]
Combining Lemma \ref{lemma4} and the local alternative (\ref{localalt}), we obtain
\begin{align*}
\frac{\left| v^\top (\hat{P} - P) \bar{z} \right|}{\sqrt{2p(k - 1)}} = o_p\left( \frac{ \max_{1 \leq j \leq k} \lambda_{j1}^{-2} \cdot p }{ \sqrt{n} } \right) \quad \frac{\left| v^\top (\hat{P} - P) v \right|}{\sqrt{2p(k - 1)}} = o_p\left( \frac{ \max_{1 \leq j \leq k} \lambda_{j1}^{-2} \cdot p }{ \sqrt{n} } \right).
\end{align*}
Furthermore, invoking Lemmas \ref{lemma2} and \ref{lemma3}, together with the martingale central limit theorem in (\cite{HallHeyde1980}) , we obtain  asymptotic normality:
\[
\frac{T^{\mathcal{P}}_{\mathrm{FLRT}} - \operatorname{tr}(Q) - v^\top Q v}{\sqrt{2( \operatorname{tr}(Q) + 2 v^\top Q v )}} \xrightarrow{d} \mathcal{N}(0, 1).
\]

We know \( \operatorname{tr}(Q) = p(k - 1) \), and the local alternative condition in (\ref{localalt}) tells that $|v^\top Qv| \leq ||v||^2=o(p)$. Noting that $T_{FLRT}=T_{FLRT}^{\mathcal{P}}+\bar{z}^\top (\hat{P} - P) \bar{z}+2v^\top (\hat{P} - P) \bar{z}+v^\top (\hat{P} - P) v$, the proof of Theorem 1 is now complete.
\end{proof}
\begin{proof}[Proof of Theorem \ref{theorem2}]

All quantities carrying a superscript \(\ast\) are understood to be \emph{conditionally} defined given the observed sample.  Conditional on the observed data, since transformations will still be performed, we have the following bootstrapped versions of key quantities: $K_n^{\ast} \;=\; \sum_{m=1}^n x_m^{\ast}$,
$\mathbb{E}^{\ast}(K_n^{\ast}) = 0$,
$\operatorname{Var}^{\ast}(K_n^{\ast}) = I_{pk}$,
$T_{\text{FLRT}}^{\ast} = \phi(K_n^{\ast},Q^{\ast})$, with
\begin{equation}
\phi(z,S) \;=\;\frac{ z^{\top} S z - p(k-1) - v^{\top} Q v }{ \sqrt{2p(k-1)} } ,
\label{phi_formula}
\end{equation}
and where \(Q^{\ast}\) uses the bootstrap plug-in estimators of the population covariance matrices.

Conditional on the sample, the multivariate Berry–Esseen theorem (\cite{Raic2019}) yields
\begin{equation}
\label{eq:BE}
\bigl|\,
\mathbb{P}^{\ast}(K_n^{\ast}\!\in A) \;-\; \mathcal{N}(0,I_{pk})(A)
\,\bigr|
\;\lesssim\;
\,p^{1/4}\,\sum_{m=1}^{n}\bigl\lVert x_m\bigr\rVert^{3}
\quad
\text{for every convex set } A\subset\mathbb{R}^{pk},
\end{equation}
where, for $a,b \in \mathbb{R}$,  $ a\lesssim b$ means $a \leq cb$ for some constant $c>0$.

Using the same argument as in Lemma~3, one obtains the stochastic approximation
\begin{equation}
\label{eq:phi-approx}
\bigl|
\phi(K_n^{\ast},Q^{\ast})-\phi(K_n^{\ast},Q)
\bigr|
\;=\;
\frac{(K_n^{\ast})^{\top}(Q^{\ast}-Q)K_n^{\ast}}{\sqrt{2p(k-1)}}
\;=\;
O_{p}\!(
c_n
),
\qquad
c_n
=\frac{ \max_{1\le j\le k}\lambda_{j1}^{-2}\,p}{\sqrt{n}},
\end{equation}
where \(\lambda_{j1}\) is the smallest eigenvalue of the \(j\)-th population covariance matrix.  
Under the rate condition stated in the theorem, \(c_n\to0\) as \(n\to\infty\).

Fix \(\varepsilon>0\).  There exists \(N\in\mathbb{N}\) such that, for all \(n>N\),
\begin{equation}
\label{eq:epsilon}
\mathbb{P}^{\ast}\!\Bigl(
\bigl|
\phi(K_n^{\ast},Q^{\ast})-\phi(K_n^{\ast},Q)
\bigr|
\le c_n
\Bigr)
\;\ge\;
1-\varepsilon.
\end{equation}

Let \(x\in\mathbb{R}\) be fixed and set the convex set
\(A_n = \{z\in\mathbb{R}^{pk} : \phi(z,Q)\le x+c_n\}\).
Using~\eqref{eq:epsilon},
\begin{align*}
        &\mathbb{P}^*(\phi(K_n^*,Q^*) \leq x)\\
        &=\mathbb{P}^*(\phi(K_n^*,Q) \leq x +\phi(K_n^*,Q) -\phi(K_n^*,Q^*))\\
        &=\mathbb{P}^*(\{\phi(K_n^*,Q) \leq x +\phi(K_n^*,Q) -\phi(K_n^*,Q^*)\} \cap \{|\phi(K_n^*,Q^*)-\phi(K_n^*,Q)| \leq c_n\})\\
        &+\mathbb{P}^*(\{\phi(K_n^*,Q) \leq x +\phi(K_n^*,Q) -\phi(K_n^*,Q^*)\} \cap \{|\phi(K_n^*,Q^*)-\phi(K_n^*,Q)| > c_n\})\\
        & \leq \mathbb{P}^*(\phi(K_n^*,Q) \leq x +c_n)+\varepsilon=
\mathbb{P}^{\ast}\!\bigl(K_n^{\ast}\in A_n\bigr)+\varepsilon.
\end{align*}
Let \(Z\sim\mathcal{N}(0,I_{pk})\) and consider \(X=\phi(Z,Q)\), where $\phi$ is defined in (\ref{phi_formula}).  Also,  let $F_X(x)$ denote the cumulative distribution function of $X$..
Then
\begin{align*}
\bigl|
\mathbb{P}^{\ast}\!\bigl(\phi(K_n^{\ast},Q^{\ast})\le x\bigr)-\Phi(x)
\bigr|
&\le
\bigl|
\mathbb{P}^{\ast}(K_n^{\ast}\in A_n)-\mathcal{N}(0,I_{pk})(A_n)
\bigr|\\
&\quad+
\bigl|
\mathcal{N}(0,I_{pk})(A_n)-\Phi(x)
\bigr|
+\varepsilon\\
&=
\underbrace{\bigl|
\mathbb{P}^{\ast}(K_n^{\ast}\in A_n)-\mathcal{N}(0,I_{pk})(A_n)
\bigr|}_{\text{(I)}}\\
&\quad+
\underbrace{\bigl|F_X(x+c_n)-F_X(x)\bigr|}_{\text{(II)}}
+
\underbrace{\bigl|F_X(x)-\Phi(x)\bigr|}_{\text{(III)}}
+\varepsilon.
\end{align*}
Term (I) is bounded by \eqref{eq:BE}.  
Term (II) is at most \(c_n\) by the Lipschitz continuity of \(F_X\).
Term (III) is \(O(p^{-1/2})\) by the classical normal approximation for non-central chi square distribution as degree of freedom or non-centrality parameter goes to infinity.
Hence
\[
\bigl|
\mathbb{P}^{\ast}\!\bigl(\phi(K_n^{\ast},Q^{\ast})\le x\bigr)-\Phi(x)
\bigr|
\;\lesssim\;
p^{1/4}\sum_{m=1}^n\lVert x_m\rVert^{3}
\;+\;
c_n
\;+\;
p^{-1/2}
\;+\;
\varepsilon.
\]

Using the tower property,
\[
\mathbb{P}\!\bigl(\phi(K_n^{\ast},Q^{\ast})\le x\bigr)
\;=\;
\mathbb{E}\bigl[
\mathbb{P}^{\ast}\!\bigl(\phi(K_n^{\ast},Q^{\ast})\le x\bigr)
\bigr],
\]
so that
\begin{align*}
\bigl|
\mathbb{P}\!\bigl(\phi(K_n^{\ast},Q^{\ast})\le x\bigr)-\Phi(x)
\bigr|
&\le
\mathbb{E}
\left[
\bigl|
\mathbb{P}^{\ast}\!\bigl(\phi(K_n^{\ast},Q^{\ast})\le x\bigr)-\Phi(x)
\bigr|
\right]\\
&\lesssim
p^{1/4}\sum_{m=1}^n\mathbb{E}\lVert x_m\rVert^{3}
\;+\;
c_n
\;+\;
p^{-1/2}
\;+\;
\varepsilon\\
&\lesssim
\frac{p^{7/4}}{\sqrt{n}}
\;+\;
c_n
\;+\;
p^{-1/2}
\;+\;
\varepsilon,
\end{align*}
which completes the proof.
\end{proof}

\section{Discussion}
In the present article, we propose a projection-based test for equality of mean functions that allows the covariance operators to differ across groups and does not require Gaussianity. The number of retained scores is permitted to grow with sample size under mild rate conditions, and the method accommodates a wide range of orthonormal bases. We also develop a bootstrap calibration for the statistic and provide diagnostics procedures. Although the dimensionality must increase more slowly than the sample size, this requirement is not overly restrictive in practice: there is a natural trade-off between the number of scores and the smoothness of the underlying functions, and our simulations show that flexible basis selection can substantially reduce the number of scores needed. While our presentation focuses on functions observed over intervals, the framework is broadly compatible with other functional domains and has the potential to extend to richer settings.

\section*{Acknowledgements} 
ATAW is grateful to the Australian Research Council for  support through grant DP220102232.

  \bibliography{bibliography.bib}

\vfill \eject

\section{Supplementaty Materials for\\ "Equivalence Test for Mean Functions from Multi-population Functional Data"\\ by Chuang Xu, Andrew T.A. Wood and Yanrong Yang}
\subsection{Proofs of Main Lemmas}
In this section, the technical proofs of all stated Lemmas are presented.
\begin{proof}[Proof of Lemma \ref{lemma2}]
    We now compute \( D_{nm}^2 \), and under the moment assumption \( \mathbb{E}[z_{jid}^4] = 3 + \Delta \), it follows that
\[
\mathbb{E}[x_m^\top Q x_m x_m^\top Q x_m] = \frac{1}{n_j^2} \left( \mathrm{tr}(Q_{j,j})^2 + 2\mathrm{tr}(Q_{j,j}^2) + \Delta \sum_{d=1}^p Q_{j,j}[d,d]^2 \right).
\]

Using this, we compute the conditional second moment:
\begin{align*}
\mathbb{E}[D_{nm}^2 | \mathcal{F}_{nm-1}] &= 
    4K_{m-1}^\top QI_{pk,m} Q K_{m-1}+4K_{m-1}^\top Q \mathbb{E}\{x_m\left(x_m^\top Q x_m-tr(I_{pk,m}Q)\right)\}
    +\dfrac{2}{n_j^2} tr(Q_{j,j}^2)\\
    &+\dfrac{\Delta}{n_j^2}\sum_{d=1}^p Q_{j,j}[d,d]^2
    +4v^\top I_{pk,m} Qv+8K_{m-1}^\top Q I_{pk,m} Qv +4v^\top Q \mathbb{E}(x_mx_m^\top Qx_m)
\end{align*}
We define
\[
\eta_n = \sum_{m=1}^n \mathbb{E}(D_{nm}^2 \mid \mathcal{F}_{nm-1}),
\]
which can be decomposed as
\begin{align*}
\eta_n &= 
4 \sum_{m=1}^n \left\{ K_{m-1}^\top Q I_{pk,m} Q K_{m-1} 
+ K_{m-1}^\top Q \, \mathbb{E}\left[ x_m (x_m^\top Q x_m - \operatorname{tr}(I_{pk,m} Q)) \right] \right\} \\
&\quad + 2 \sum_{j=1}^k \frac{\operatorname{tr}(Q_{j,j}^2)}{n_j} 
+ \Delta \sum_{j=1}^k \frac{\sum_{d=1}^p Q_{j,j}[d,d]^2}{n_j} \\
&\quad + 4 \sum_{m=1}^n v^\top I_{pk,m} Q v 
+ 8 \sum_{m=1}^n K_{m-1}^\top Q I_{pk,m} Q v 
+ 4 \sum_{m=1}^n v^\top Q \, \mathbb{E}(x_m x_m^\top Q x_m) \\
&= \eta_{n1} + \eta_{n2},
\end{align*}
where we define the two components as
\begin{align*}
\eta_{n1} &= 4 \sum_{m=1}^n \left\{ K_{m-1}^\top Q I_{pk,m} Q K_{m-1} 
+ K_{m-1}^\top Q \, \mathbb{E}\left[ x_m (x_m^\top Q x_m - \operatorname{tr}(I_{pk,m} Q)) \right] \right\} \\
&\quad + 2 \sum_{j=1}^k \frac{\operatorname{tr}(Q_{j,j}^2)}{n_j} 
+ \Delta \sum_{j=1}^k \frac{\sum_{d=1}^p Q_{j,j}[d,d]^2}{n_j}, \\
\eta_{n2} &= 4 \sum_{m=1}^n v^\top I_{pk,m} Q v 
+ 8 \sum_{m=1}^n K_{m-1}^\top Q I_{pk,m} Q v 
+ 4 \sum_{m=1}^n v^\top Q \, \mathbb{E}(x_m x_m^\top Q x_m).
\end{align*}

To compute \( \mathbb{E}(\eta_n) \), we analyze \( \mathbb{E}(\eta_{n1}) \) and \( \mathbb{E}(\eta_{n2}) \) separately. We begin with
\[
\mathbb{E}(\eta_{n1}) = 4 \sum_{m=1}^n \mathbb{E}(K_{m-1}^\top Q I_{pk,m} Q K_{m-1}) 
+ 2 \sum_{j=1}^k \frac{\operatorname{tr}(Q_{j,j}^2)}{n_j} 
+ \Delta \sum_{j=1}^k \frac{\sum_{d=1}^p Q_{j,j}[d,d]^2}{n_j}.
\]

By independence and properties of conditional expectation, we obtain:
\begin{align*}
\mathbb{E}(K_{m-1}^\top Q I_{pk,m} Q K_{m-1}) 
&= \sum_{s=1}^{m-1} \mathbb{E}(x_s^\top Q I_{pk,m} Q x_s) 
= \sum_{s=1}^{m-1} \operatorname{tr}(I_{pk,s} Q I_{pk,m} Q) \\
&= \operatorname{tr}\left( \left( \sum_{s=1}^{m-1} I_{pk,s} \right) Q I_{pk,m} Q \right).
\end{align*}

Let us evaluate this for each group \( j \). For instance, when \( m \in \{2, \ldots, n_1\} \), we get:
\[
\operatorname{tr}\left( \left( \sum_{s=1}^{m-1} I_{pk,s} \right) Q I_{pk,m} Q \right)
= \frac{m-1}{n_1^2} \operatorname{tr}(Q_{1,1}^2),
\]
and summing over \( m \leq n_1 \), we obtain:
\[
\sum_{m=1}^{n_1} \mathbb{E}(K_{m-1}^\top Q I_{pk,m} Q K_{m-1}) 
= \frac{n_1 - 1}{2 n_1} \operatorname{tr}(Q_{1,1}^2).
\]

Similarly, for \( m \in \{n_1 + 1, \ldots, n_1 + n_2\} \), we find:
\[
\operatorname{tr}\left( \left( \sum_{s=1}^{m-1} I_{pk,s} \right) Q I_{pk,m} Q \right)
= \frac{1}{n_2} \operatorname{tr}(Q_{1,2}^2) + \frac{m - n_1 - 1}{n_2^2} \operatorname{tr}(Q_{2,2}^2),
\]
and summing over this range gives:
\[
\operatorname{tr}(Q_{1,2}^2) + \frac{n_2 - 1}{2 n_2} \operatorname{tr}(Q_{2,2}^2).
\]

Proceeding inductively, we obtain the full expression:
\[
\sum_{m=1}^n \mathbb{E}(K_{m-1}^\top Q I_{pk,m} Q K_{m-1}) 
= \sum_{j=1}^k \sum_{i=1}^{j-1} \operatorname{tr}(Q_{i,j}^2) 
+ \sum_{j=1}^k \frac{n_j - 1}{2 n_j} \operatorname{tr}(Q_{j,j}^2).
\]

Noting that \( Q_{i,j} = Q_{j,i} \) due to symmetry of \( Q \), we simplify:
\begin{align*}
\mathbb{E}(\eta_{n1}) 
&= 4 \sum_{j=1}^k \sum_{i=1}^{j-1} \operatorname{tr}(Q_{i,j}^2) 
+ 2 \sum_{j=1}^k \frac{n_j - 1}{n_j} \operatorname{tr}(Q_{j,j}^2) 
+ 2 \sum_{j=1}^k \frac{\operatorname{tr}(Q_{j,j}^2)}{n_j} 
+ \Delta \sum_{j=1}^k \frac{\sum_d Q_{j,j}[d,d]^2}{n_j} \\
&= 2 \sum_{i,j} \operatorname{tr}(Q_{i,j}^2) 
+ \Delta \sum_{j=1}^k \frac{\sum_d Q_{j,j}[d,d]^2}{n_j} \\
&= 2 \operatorname{tr}(Q^2) + \Delta \sum_{j=1}^k \frac{\sum_d Q_{j,j}[d,d]^2}{n_j} \\
&= 2p(k-1) + \Delta \sum_{j=1}^k \frac{\sum_d Q_{j,j}[d,d]^2}{n_j}.
\end{align*}

We now compute:
\begin{align*}
\mathbb{E}(\eta_{n2}) 
&= 4 v^\top \left( \sum_{m=1}^n I_{pk,m} \right) Q v 
+ 4 \sum_{m=1}^n v^\top Q \, \mathbb{E}(x_m x_m^\top Q x_m) \\
&= 4 v^\top Q v + 4 v^\top Q \left( \sum_{m=1}^n \mathbb{E}(x_m x_m^\top Q x_m) \right).
\end{align*}

Combining the two parts:
\begin{align*}
\mathbb{E}(\eta_n) 
&= 2p(k-1) + 4 v^\top Q v 
+ 4 v^\top Q \left( \sum_{m=1}^n \mathbb{E}(x_m x_m^\top Q x_m) \right) 
+ \Delta \sum_{j=1}^k \frac{\sum_d Q_{j,j}[d,d]^2}{n_j}.
\end{align*}

We now bound the third term. Observe that:
\[
\left\| \mathbb{E}(x_m x_m^\top Q x_m) \right\|_2 \leq \frac{\| Q_{j(m), j(m)} \|_F}{n_{j(m)}^{3/2}},
\]
hence,
\[
v^\top Q \left( \sum_{m=1}^n \mathbb{E}(x_m x_m^\top Q x_m) \right)
\lesssim \|v_d\|_2 ||Q||_2\|Q\|_F \lesssim \|v_d\|_2 \sqrt{p}.
\]
Also,
\[
\Delta \sum_{j=1}^k \frac{\sum_d Q_{j,j}[d,d]^2}{n_j} \lesssim \frac{\operatorname{tr}(Q^2)}{n} = \frac{tr(Q)}{n}.
\]

Under the alternative hypothesis and the assumed rate, we conclude:
\[
\mathbb{E}(\eta_n) = 2(tr(Q) + 2 v^\top Q v)(1 + o(1))=\sigma_n^2(1 + o(1)).
\]
We aim to show that \( \dfrac{\operatorname{Var}(\eta_n)}{\sigma_n^4} \to 0 \) as \( n \to \infty \). To do this, we consider the decomposition
\[
\eta_n = \eta_{n1} + \eta_{n2},
\]
and show that both \( \dfrac{\operatorname{Var}(\eta_{n1})}{\sigma_n^4} \to 0 \) and \( \dfrac{\operatorname{Var}(\eta_{n2})}{\sigma_n^4} \to 0 \).

Recall that
\begin{align*}
\eta_{n1}
&= 4 \sum_{m=1}^n \left\{ 
K_{m-1}^\top Q I_{pk,m} Q K_{m-1} 
+ K_{m-1}^\top Q \, \mathbb{E}\left[ x_m \left( x_m^\top Q x_m \right) \right]
\right\} 
+ c \\
&= 4L_1 + 4L_2 + c,
\end{align*}
where we define:
\begin{align*}
L_1 &= \sum_{m=1}^n K_{m-1}^\top Q I_{pk,m} Q K_{m-1}, \\
L_2 &= \sum_{m=1}^n K_{m-1}^\top Q \, \mathbb{E}\left[ x_m \left( x_m^\top Q x_m \right) \right].
\end{align*}

We expand \( L_1 \) as:
\begin{align*}
L_1 
&= \sum_{m=1}^n \left( \sum_{s=1}^{m-1} x_s^\top Q I_{pk,m} Q x_s 
+ 2 \sum_{d=1}^{m-1} \sum_{t=1}^{d-1} x_d^\top Q I_{pk,m} Q x_t \right) \\
&= \sum_{s=1}^{n-1} x_s^\top \left( \sum_{m=s+1}^{n} Q I_{pk,m} Q \right) x_s 
+ 2 \sum_{d=1}^{n-1} \sum_{t=1}^{d-1} x_d^\top \left( \sum_{m=d+1}^{n} Q I_{pk,m} Q \right) x_t \\
&= L_{11} + 2L_{12}.
\end{align*}

Using Lemma 2.3 from \cite{Bai1998}, we have:
\begin{align*}
\mathbb{E}(L_{11}^2) 
&= \mathbb{E} \left( \sum_{s=1}^{n-1} \sum_{m=s+1}^{n} \frac{1}{n_s n_m} 
z_{j(s)i(s)}^\top Q_{j(s),j(m)}^2 z_{j(s)i(s)} \right)^2 \\
&\lesssim \frac{ \left( \sum_{s=1}^{n-1} \sum_{m=s+1}^n \operatorname{tr}(Q_{j(s),j(m)}^2) \right)^2 + n^3 \operatorname{tr}(Q^4) }{n^4} 
\lesssim \frac{ \operatorname{tr}(Q)^2 }{n}.
\end{align*}

Similarly, applying Lemma 2.1 from \cite{Bai1998} to \( L_{12} \) (a martingale difference form), we get:
\begin{align*}
\mathbb{E}(L_{12}^2) 
&= \mathbb{E} \left( \sum_{d=1}^{n-1} \sum_{t=1}^{d-1} x_d^\top 
\left( \sum_{m=d+1}^{n} Q I_{pk,m} Q \right) x_t \right)^2 \\
&\lesssim \frac{ \sum Q_{j(d),j(m)} Q_{j(m),j(s)} Q_{j(m),j(s)} Q_{j(m),j(d)} }{n^4} 
\lesssim \frac{ \operatorname{tr}(Q^4) }{n}
\lesssim \frac{ \operatorname{tr}(Q)^2 }{n}.
\end{align*}

We now control the second term:
\begin{align*}
\mathbb{E}(L_2^2) 
&= \mathbb{E} \left( \sum_{s=1}^{n-1} x_s^\top 
\left( \sum_{m=s+1}^n Q \, \mathbb{E}[x_m (x_m^\top Q x_m)] \right) \right)^2 \\
&= \sum_{s=1}^{n-1} \mathbb{E} \left( x_s^\top 
\left( \sum_{m=s+1}^n Q \, \mathbb{E}[x_m (x_m^\top Q x_m)] \right) \right)^2 \\
&\lesssim \frac{1}{n^4} \sum_{s=1}^{n-1} \mathbb{E} \left( z_{j(s), i(s)}^\top 
\left( \sum_{m=s+1}^n Q_{j(s), j(m)} \operatorname{diag}(Q_{j(m), j(m)}) \right) \right)^2 \\
&\lesssim \frac{ \sum_{s=1}^{n-1} 
\left\| \sum_{m=s+1}^n Q_{j(s), j(m)} \operatorname{diag}(Q_{j(m), j(m)}) \right\|^2 }{n^4} \\
&\lesssim \frac{ \sum_{s=1}^{n-1} \sum_{m=s+1}^n \operatorname{tr}(Q_{j(s),j(m)}^2) \operatorname{tr}(Q_{j(m),j(m)}^2) }{n^4} 
\lesssim \frac{ \operatorname{tr}(Q^2)^2 }{n}.
\end{align*}

From the decomposition
\[
\eta_{n2} = 8 \sum_{m=1}^n K_{m-1}^\top Q I_{pk,m} Q v + c,
\]
we control the variance via a similar technique:
\begin{align*}
\mathbb{E} \left( \sum_{m=1}^n K_{m-1}^\top Q I_{pk,m} Q v \right)^2
&= \mathbb{E} \left( \sum_{s=1}^{n-1} x_s^\top 
\left( \sum_{m=s+1}^n Q I_{pk,m} Q v \right) \right)^2 \\
&\lesssim n \cdot v^\top Q v.
\end{align*}

Combining all results, we obtain:
\[
\dfrac{\operatorname{Var}(\eta_{n1})}{\sigma_n^4} \lesssim \frac{1}{n}, 
\qquad 
\dfrac{\operatorname{Var}(\eta_{n2})}{\sigma_n^4} \lesssim \dfrac{1}{\sigma^2}.
\]
It follows that:
\[
\dfrac{\operatorname{Var}(\eta_n)}{\sigma_n^4} = o(1).
\]
This completes the proof of Lemma 1.
\end{proof}
\begin{proof}[Proof of Lemma \ref{lemma3}]
    The conclusion of the lemma is true if we show
    \[\dfrac{1}{\sigma_n^4} \sum_{m=1}^n \mathbb{E}( D^4_{nm}) \to 0.\]
    
    We now derive an upper bound on the fourth moment sum \( \sum_{m=1}^n \mathbb{E}(D_{nm}^4) \).

Recall from earlier that
\begin{align*}
D_{nm}^4 
&= \left\{ 2 x_m^\top Q K_{m-1} + \left( x_m^\top Q x_m - \operatorname{tr}(I_{pk,m} Q) \right) + 2 v^\top Q x_m \right\}^4 \\
&\lesssim (x_m^\top Q K_{m-1})^4 
+ \left(x_m^\top Q x_m - \operatorname{tr}(I_{pk,m} Q) \right)^4 
+ (2 v^\top Q x_m)^4 \\
&=: A_{m1}^4 + A_{m2}^4 + A_{m3}^4.
\end{align*}

We analyze \( A_{m1}^4 \) and \( A_{m2}^4 \), since direct computations can show
\[\sum_{m=1}^n \mathbb{E}(A_{m3}^4) \lesssim \dfrac{(v^\top Q v)^2}{n^2}.\]

Using Lemma 2.7 from \cite{Bai1998} and assuming finite 8th moments, we have:
\begin{align*}
\mathbb{E}(A_{m2}^4) 
&= \mathbb{E}\left[ \left( x_m^\top Q x_m - \operatorname{tr}(I_{pk,m} Q) \right)^4 \right] \\
&= \frac{1}{n_{j(m)}^4} 
\mathbb{E}\left[ \left( z_{j(m),i(m)}^\top Q_{j(m),j(m)} z_{j(m),i(m)} 
- \operatorname{tr}(Q_{j(m),j(m)}) \right)^4 \right] \\
&\lesssim \frac{1}{n_{j(m)}^4} \operatorname{tr}(Q_{j(m),j(m)}^2)^2.
\end{align*}

Since \( n_j / n \to c_j > 0 \), we sum over \( m \) to get:
\begin{align*}
\sum_{m=1}^n \mathbb{E}(A_{m2}^4) 
&\lesssim \sum_{j=1}^k \frac{1}{n_j^3} \operatorname{tr}(Q_{j,j}^2)^2 
\lesssim \frac{1}{n^3} \sum_{j=1}^k \operatorname{tr}(Q_{j,j}^2)^2 
\lesssim \frac{1}{n^3} \operatorname{tr}(Q^2)^2 
= \frac{p^2(k-1)^2}{n^3}.
\end{align*}

We treat \( x_m^\top Q x_s \) for \( s < m \) as a martingale difference sequence with respect to the filtration 
\( \mathcal{F}_s = \sigma(x_1, \dots, x_s) \). Applying Lemma 2.1 from \cite{Bai1998}, we get:
\begin{align*}
\mathbb{E}(A_{m1}^4) 
&\leq \left( \sum_{s=1}^{m-1} \mathbb{E}(x_m^\top Q x_s x_m^\top Q x_s) \right)^2 
+ \sum_{s=1}^{m-1} \mathbb{E}\left( x_m^\top Q x_s \right)^4.
\end{align*}

The first term involves:
\begin{align*}
\left( \sum_{s=1}^{m-1} \mathbb{E}(x_m^\top Q x_s x_m^\top Q x_s) \right)^2 
&\lesssim \frac{\left( \sum_{s=1}^{m-1} \operatorname{tr}(Q_{j(m),j(s)}^2) \right)^2}{n^4}.
\end{align*}

Summing over \( m \), we obtain:
\begin{align*}
\sum_{m=1}^n \left( \sum_{s=1}^{m-1} \mathbb{E}(x_m^\top Q x_s x_m^\top Q x_s) \right)^2 
&\lesssim \frac{ \left( \sum_{m=1}^n \sum_{s=1}^{m-1} \operatorname{tr}(Q_{j(m),j(s)}^2) \right)^2 }{n^4} \\
&\lesssim \frac{n^3 \operatorname{tr}(Q^2)^2}{n^4} = \frac{\operatorname{tr}(Q)^2}{n}.
\end{align*}

Using Lemmas 2.1 and 2.3 from \cite{Bai1998}, and denoting \( Q_{j,i}[l,t] \) as the \((l,t)\)-entry of \( Q_{j,i} \), we estimate:
\begin{align*}
\mathbb{E}\left( x_m^\top Q x_s \right)^4 
&\lesssim \frac{1}{n^4} \mathbb{E} \left( \sum_{l=1}^p z_{j(m),i(m),l} \sum_{t=1}^p Q_{j(m),j(s)}[l,t] z_{j(s),i(s),t} \right)^4 \\
&\lesssim \frac{ \operatorname{tr}(Q_{j(m),j(s)}^2) + \operatorname{tr}(Q_{j(m),j(s)}^4) }{n^4}.
\end{align*}

Therefore,
\begin{align*}
\sum_{m=1}^n \sum_{s=1}^{m-1} \mathbb{E}\left( x_m^\top Q x_s \right)^4 
&\lesssim \frac{n^3 \operatorname{tr}(Q)^2}{n^4} = \frac{\operatorname{tr}(Q)^2}{n}.
\end{align*}

Combining all the bounds, we obtain:
\[
\dfrac{1}{\sigma_n^4}\sum_{m=1}^n \mathbb{E}(D_{nm}^4) 
\lesssim \dfrac{1}{n}
\to 0 \quad \text{as } n \to \infty.
\]
Thus, the fourth moment of the martingale difference array is asymptotically negligible under the assumed conditions.

\end{proof}
\begin{proof}[Proof of Lemma \ref{lemma4}]
    We now compare the block matrices \( P \) and \( \hat{P} \). Since \( k \) is fixed, it suffices to examine the convergence of each sub-block of these matrices individually. Each sub-block can be expressed as:
\[
n_s^{1/2} \Sigma_s^{-1/2} \left( \sum_{j=1}^k n_j \Sigma_j^{-1} \right)^{-1} n_l^{1/2} \Sigma_l^{-1/2}
- 
n_s^{1/2} \hat{\Sigma}_s^{-1/2} \left( \sum_{j=1}^k n_j \hat{\Sigma}_j^{-1} \right)^{-1} n_l^{1/2} \hat{\Sigma}_l^{-1/2}.
\]

Without loss of generality, we can factor out the \( n_j \) terms to focus on the main object of interest:
\[
\Sigma_s^{-1/2} \left( \sum_{j=1}^k \Sigma_j^{-1} \right)^{-1} \Sigma_l^{-1/2}
-
\hat{\Sigma}_s^{-1/2} \left( \sum_{j=1}^k \hat{\Sigma}_j^{-1} \right)^{-1} \hat{\Sigma}_l^{-1/2}.
\]

We begin by analyzing the inverse of the sample covariance matrix. Using a matrix expansion, we write:
\[
\hat{\Sigma}_j^{-1} = \Sigma_j^{-1/2} \left( \Sigma_j^{-1/2} \hat{\Sigma}_j \Sigma_j^{-1/2} - I + I \right)^{-1} \Sigma_j^{-1/2}.
\]
Define
\[
V_j = \Sigma_j^{-1/2} \hat{\Sigma}_j \Sigma_j^{-1/2} - I,
\]
so that we may expand:
\begin{align*}
\hat{\Sigma}_j^{-1} &= \Sigma_j^{-1} - \Sigma_j^{-1/2} V_j \Sigma_j^{-1/2} + \Sigma_j^{-1/2} V_j^2 \Sigma_j^{-1/2} - \cdots, \\
\sum_{j=1}^k \hat{\Sigma}_j^{-1} &= \sum_{j=1}^k \Sigma_j^{-1} - \sum_{j=1}^k \Sigma_j^{-1/2} V_j \Sigma_j^{-1/2} + \sum_{j=1}^k \Sigma_j^{-1/2} V_j^2 \Sigma_j^{-1/2} - \cdots \\
&= A + B,
\end{align*}
where \( A = \sum_{j=1}^k \Sigma_j^{-1} \) and \( B \) collects the remaining error terms.

One can easily show that 
\[
\operatorname{tr}(V_j^2) = O_p\left( \frac{p}{n} \right).
\]

Using a Neumann series expansion and bounding the spectral norm, we obtain:
\begin{align*}
\left( \sum_{j=1}^k \hat{\Sigma}_j^{-1} \right)^{-1} 
&= A^{-1} - A^{-1} B A^{-1} + A^{-1} B A^{-1} B A^{-1} - \cdots \\
&= A^{-1} + D,
\end{align*}
where the error term \( D \) satisfies:
\begin{align*}
\| D \| &\leq \| A^{-1} \| \cdot \left\| B A^{-1} + B A^{-1} B A^{-1} + \cdots \right\| \\
&\leq \| A^{-1} \|^2 \| B \| \cdot \sum_{i=0}^\infty \left( \| B \| \| A^{-1} \| \right)^i \\
&= \frac{ \| A^{-1} \|^2 \| B \| }{ 1 - \| A^{-1} \| \| B \| }.
\end{align*}

We now analyze \( \| B \| \). Using the submultiplicativity of the operator norm:
\begin{align*}
\| B \| 
&\leq \sum_{j=1}^k \left\| \Sigma_j^{-1/2} V_j \Sigma_j^{-1/2} \right\| 
\leq k \cdot \max_{1 \leq j \leq k} \| \Sigma_j^{-1} \| \cdot \sum_{i=1}^\infty \| V_j \|^i \\
&= \frac{ k \cdot \max_{1 \leq j \leq k} \lambda_{j1}^{-1} \cdot \| V_j \| }{ 1 - \| V_j \| }.
\end{align*}

For the matrix \( A^{-1} \), we have:
\begin{align*}
\| A^{-1} \| 
&= \frac{1}{ \lambda_{\min}(A) } 
\leq \frac{1}{ \sum_{j=1}^k \lambda_{\min}(\Sigma_j^{-1}) }
= \frac{1}{ \sum_{j=1}^k \| \Sigma_j \|^{-1} } \\
&= \frac{1}{ k \cdot \min_{1 \leq j \leq k} \| \Sigma_j \|^{-1} } 
= \frac{ \max_{1 \leq j \leq k} \| \Sigma_j \| }{ k }.
\end{align*}

Assuming the spectral norms of \( \Sigma_j \) are uniformly bounded above, it follows that \( \| A^{-1} \| \) is bounded.

We conclude:
\[
\| B \| = O_p\left( \frac{ \max_{1 \leq j \leq k} \lambda_{j1}^{-1}\sqrt{p} }{ \sqrt{n} } \right).
\]

Returning to our original target expression:
\begin{align*}
&\hat{\Sigma}_s^{-1/2} \left( \sum_{j=1}^k \hat{\Sigma}_j^{-1} \right)^{-1} \hat{\Sigma}_l^{-1/2} \\
= &(\hat{\Sigma}_s^{-1/2} - \Sigma_s^{-1/2} + \Sigma_s^{-1/2})
\left( \sum_{j=1}^k \hat{\Sigma}_j^{-1} \right)^{-1}
(\hat{\Sigma}_l^{-1/2} - \Sigma_l^{-1/2} + \Sigma_l^{-1/2}) \\
= &(\hat{\Sigma}_s^{-1/2} - \Sigma_s^{-1/2}) \left( \sum_{j=1}^k \hat{\Sigma}_j^{-1} \right)^{-1} (\hat{\Sigma}_l^{-1/2} - \Sigma_l^{-1/2}) \\
&+ (\hat{\Sigma}_s^{-1/2} - \Sigma_s^{-1/2}) \left( \sum_{j=1}^k \hat{\Sigma}_j^{-1} \right)^{-1} \Sigma_l^{-1/2} \\
&+ \Sigma_s^{-1/2} \left( \sum_{j=1}^k \hat{\Sigma}_j^{-1} \right)^{-1} (\hat{\Sigma}_l^{-1/2} - \Sigma_l^{-1/2}) \\
&+ \Sigma_s^{-1/2} \left( \sum_{j=1}^k \hat{\Sigma}_j^{-1} \right)^{-1} \Sigma_l^{-1/2}.
\end{align*}

By Theorem 6.2 in \cite{higham2008functions} and the bounded trace of population covariance matrices, we can see that \( \| \hat{\Sigma}_j^{-1/2} - \Sigma_j^{-1/2} \| = O_p(\frac{\lambda_{j1}^{-1/2}}{n}) \), thus the dominant term is:
\[
\Sigma_s^{-1/2} \left( \sum_{j=1}^k \hat{\Sigma}_j^{-1} \right)^{-1} \Sigma_l^{-1/2}.
\]

We now bound the deviation between the estimated and true blocks:
\begin{align*}
&\left\| \Sigma_s^{-1/2} \left( \sum_{j=1}^k \Sigma_j^{-1} \right)^{-1} \Sigma_l^{-1/2} 
- \hat{\Sigma}_s^{-1/2} \left( \sum_{j=1}^k \hat{\Sigma}_j^{-1} \right)^{-1} \hat{\Sigma}_l^{-1/2} \right\| \\
&\leq \left\| \Sigma_s^{-1/2} \left( \sum_{j=1}^k \Sigma_j^{-1} \right)^{-1} \Sigma_l^{-1/2}
- \Sigma_s^{-1/2} \left( \sum_{j=1}^k \hat{\Sigma}_j^{-1} \right)^{-1} \Sigma_l^{-1/2} \right\| \\
&= \left\| \Sigma_s^{-1/2} D \Sigma_l^{-1/2} \right\| 
\leq \max_{1 \leq j \leq k} \lambda_{j1}^{-1} \cdot \| D \| \\
&= O_p\left( \frac{ \max_{1 \leq j \leq k} \lambda_{j1}^{-2}\sqrt{p} }{ \sqrt{n} } \right).
\end{align*}
This completes the proof of the Lemma.
\end{proof}

\subsection{Imperfect sampling} \label{imperfect sampling}
Let $X_{ji}(\omega) \in L^2([0,1])$ be a random function, $\epsilon_{jil}$ be the measurement error with mean zero and variance $\sigma^2$, which are independent from each other. We observe the following:
\[ Y_{jil}=X_{ji}(t_{l})+\epsilon_{jil}, \quad 1 \leq j \leq k, \quad1 \leq i \leq n_j, \quad 1 \leq l \leq m,\]
where $\mathbb{E}(X_{ji}(t))=\mu_j(t)$.
Now let us consider the calculation of a single score given by the $L^2$ orthonormal basis $\{b_q(t)\}_{q=1}^\infty$. We aim to estimate the true score $\langle \mu_j,b_q\rangle$ due to the functional nature that smooth functions can be well approximated by several scores.\\
Let us consider using some easy quadratures for example the trapezium rules. The following can be done:
\begin{align*}
    N_{ji,q}=&\sum_{l=1}^{m-1} \frac{(t_{l+1}-t_{l})(Y_{ji(l+1)}b_q(t_{l+1})+Y_{ji(l)}b_q(t_{l}))}{2}\\
    =&\sum_{l=1}^{m-1} \frac{(t_{l+1}-t_{l})(X_{ji}(t_{l+1})b_q(t_{l+1})+X_{ji}(t_l)b_q(t_{l}))}{2}\\
    +&\sum_{l=1}^{m-1} \frac{(t_{l+1}-t_{l})(\epsilon_{ji(l+1)}b_q(t_{l+1})+\epsilon_{ji(l)}b_q(t_{l}))}{2}\\
    =&M_{ji,q}+E_{ji,q}
\end{align*}
If we are after an estimator for $\langle \mu_j,b_q\rangle$, then the most naive one to work with is then $\bar{N}_{j,q}= \frac{1}{n}\sum_{i=1}^n N_{ji,q}$, let us first introduce another notation
\[ I_{ji,q}=\int X_{ji}(t)b_q(t)dt.\]
The naive estimator for $\langle \mu_j,b_q \rangle$ is $\bar{N}_{j,q}$. It can be analysed as follows
\begin{align*}
    \bar{N}_{j,q}&= \frac{1}{n}\sum_{i=1}^n N_{ji,q}=\frac{1}{n}\sum_{i=1}^n M_{ji,q}+\frac{1}{n}\sum_{i=1}^n E_{ji,q}\\
    &=\frac{1}{n}\sum_{i=1}^n I_{ji,q}+\frac{1}{n}\sum_{i=1}^n (M_{ji,q}-I_{ji,q})+\frac{1}{n}\sum_{i=1}^n E_{ji,q}.
\end{align*}
The first term is $\frac{1}{n}\sum_{i=1}^n I_{ji,q}=\langle \bar{X}_{j},b_q\rangle$, direct result shows that $$\frac{1}{n}\sum_{i=1}^n I_{ji,q}-\langle \mu_j,b_q\rangle=O_p(\frac{1}{\sqrt{n}}).$$ The second term can be further analysed using the error bound of trapezium rule, 
\begin{align*}
    (M_{ji,q}-I_{ji,q})&=\sum_{l=1}^{m-1} \frac{(t_{l+1}-t_{l})(X_{ji}(t_{l+1})b_q(t_{l+1})+X_{ji}(t_l)b_q(t_{l}))}{2}-\int X_{ji}(t)b_q(t)dt\\
    &=O_p(\frac{1}{m^2}),
\end{align*}
therefore we get $\frac{1}{n}\sum_{i=1}^n (M_{ji,q}-I_{ji,q})=O_p(\frac{1}{m^2})$.
The basis functions are bounded functions, therefore we can analyse the third term as
\begin{align*}
    E_{ji,q}&=\sum_{l=1}^{m-1} \frac{(t_{l+1}-t_{l})(\epsilon_{ji(l+1)}b_q(t_{l+1})+\epsilon_{ji(l)}b_q(t_{l}))}{2}\\
    &\lesssim \sum_{l=1}^{m-1} \frac{(t_{l+1}-t_{l})(\epsilon_{ji(l+1)}+\epsilon_{ji(l)})}{2}\\
    &\lesssim \max_{l=1,...,m-1}(t_{l+1}-t_l) \sum_{l=1}^{m} \epsilon_{jil}.
\end{align*}
The order for the third term can be calculated as 
\begin{align*}
    \frac{1}{n}\sum_{i=1}^n E_{ji,q}&= \frac{\max_{l=1,...,m-1}(t_{l+1}-t_l)}{n} \sum_{i=1}^n\sum_{l=1}^{m} \epsilon_{jil}\\
    &=O_p(\frac{\max_{l=1,...,m-1}(t_{l+1}-t_l)\sqrt{m}}{\sqrt{n}}).
\end{align*}
We can see that
\[\bar{N}_{j,q}-\langle \mu_j,b_q\rangle=O_p(\frac{1}{\sqrt{n}})+O_p(\frac{1}{m^2})+O_p(\frac{\max_{l=1,...,m-1}(t_{l+1}-t_l)\sqrt{m}}{\sqrt{n}}).\]
It should be noted that better applicable quadrature would improve the rate of convergence. If one assumes common design, then the third order becomes $O_p(\frac{1}{\sqrt{nm}})$.
\subsection{Additional Numerical Results}
In this section, a few additional tables are presented for the purpose of simulation and real data analysis.
\begin{table}[H]
    \centering
    \begin{tabular}{c c}
        \toprule
        \(\nu\) & Rejection Proportion \\
        \midrule
        0.5  & 0.0672 \\
        1.0  & 0.0732 \\
        1.5  & 0.0736 \\
        2.0  & 0.0744 \\
        5.0  & 0.0765 \\
        10.0 & 0.0741 \\
        \bottomrule
    \end{tabular}
    \caption{Size of the Horváth testing procedure for various \(\nu\) values, where \(\nu\) represents the smoothness parameter. The data is generated under the same set up as Table \ref{tab:combined_single}}
    \label{tab:table3}
\end{table}
\begin{table}[H]
    \centering
    \begin{tabular}{c c c c c c c c}
        \toprule
        c Value & TLRT Haar & TLRT Fourier & T2 Haar & T2 Fourier & CS & L2b & Fb \\
        \midrule
         0.00  & 0.0508 & 0.0524 & 0.0480 & 0.0536 & 0.0526 & 0.0524 & 0.0452 \\
         0.107 & 0.0596 & 0.0614 & 0.0628 & 0.0612 & 0.0640 & 0.0656 & 0.0578 \\
         0.214 & 0.0820 & 0.0910 & 0.0830 & 0.0888 & 0.0870 & 0.0876 & 0.0762 \\
         0.321 & 0.1440 & 0.1650 & 0.1480 & 0.1540 & 0.1460 & 0.1470 & 0.1310 \\
         0.429 & 0.2250 & 0.2650 & 0.2320 & 0.2400 & 0.2420 & 0.2420 & 0.2230 \\
         0.536 & 0.3260 & 0.3920 & 0.3300 & 0.3640 & 0.3530 & 0.3530 & 0.3280 \\
         0.643 & 0.4620 & 0.5080 & 0.4770 & 0.4680 & 0.4780 & 0.4800 & 0.4490 \\
         0.750 & 0.6090 & 0.6610 & 0.6160 & 0.6270 & 0.6520 & 0.6520 & 0.6200 \\
         0.857 & 0.7320 & 0.7920 & 0.7380 & 0.7580 & 0.7910 & 0.7890 & 0.7640 \\
         0.964 & 0.8300 & 0.8750 & 0.8350 & 0.8510 & 0.8870 & 0.8860 & 0.8650 \\
         1.070 & 0.9090 & 0.9370 & 0.9090 & 0.9210 & 0.9500 & 0.9510 & 0.9400 \\
         1.180 & 0.9470 & 0.9760 & 0.9510 & 0.9680 & 0.9850 & 0.9850 & 0.9810 \\
         1.290 & 0.9790 & 0.9920 & 0.9790 & 0.9890 & 0.9950 & 0.9950 & 0.9940 \\
        \bottomrule
    \end{tabular}
    \caption{Proportion of rejections for different methods at \(p = 2\) across various \(c\) values.}
    \label{tab:table4}
\end{table}
\begin{table}[H]
  \centering
  \caption{Bootstrapped p-value summary under $H_{0}\!: \mu_1=\mu_2=\mu_{3}=\mu_{4}$.}
  \label{tab:tlrt-split supp}

  \begin{tabular}{%
      c
      >{\raggedleft\arraybackslash}p{2.9cm}
      >{\raggedleft\arraybackslash}p{1.6cm}
      @{\hspace{2em}}
      c
      >{\raggedleft\arraybackslash}p{2.9cm}
      >{\raggedleft\arraybackslash}p{1.6cm}}
    \toprule
      & \multicolumn{2}{c}{\textbf{Fourier basis}}
      & \multicolumn{2}{c}{\textbf{Haar basis}} \\[2pt]
      \cmidrule(lr){2-3} \cmidrule(lr){4-6}
      \(p\) & Test statistic & bootstrap \(p\)-value
      & \(p\) & Test statistic & bootstrap \(p\)-value \\
    \midrule
       6 & \makecell[r]{635.7618\\ {\scriptsize p-value = 0}} & 0.000
         &  6 & \makecell[r]{350.3186\\ {\scriptsize p-value = 0}} & 0.000 \\
       7 & \makecell[r]{700.9740\\ {\scriptsize p-value = 0}} & 0.000
         &  7 & \makecell[r]{700.8406\\ {\scriptsize p-value = 0}} & 0.000 \\
       8 & \makecell[r]{755.2110\\ {\scriptsize p-value = 0}} & 0.000
         &  8 & \makecell[r]{808.7973\\ {\scriptsize p-value = 0}} & 0.000 \\
       9 & \makecell[r]{899.0088\\ {\scriptsize p-value = 0}} & 0.000
         &  9 & \makecell[r]{817.1826\\ {\scriptsize p-value = 0}} & 0.000 \\
    \bottomrule
  \end{tabular}
\end{table}
\end{document}